\documentclass[french,english]{article}
\usepackage[T1]{fontenc}
\usepackage[utf8]{inputenc}
\usepackage[letterpaper]{geometry}
\usepackage{amsthm}
\usepackage{amsmath}
\usepackage{graphicx}

\makeatletter

\providecommand{\tabularnewline}{\\}

\theoremstyle{plain}
\newtheorem{thm}{Theorem}
\theoremstyle{plain}
\newtheorem{fact}[thm]{Fact}
 \ifx\proof\undefined\
   \newenvironment{proof}[1][\proofname]{\par
     \normalfont\topsep6\p@\@plus6\p@\relax
     \trivlist
     \itemindent\parindent
     \item[\hskip\labelsep
           \scshape
       #1]\ignorespaces
   }{%
     \endtrivlist\@endpefalse
   }
   \providecommand{\proofname}{Proof}
 \fi
\theoremstyle{plain}
\newtheorem{conjecture}[thm]{Conjecture}
\theoremstyle{remark}
\newtheorem{rem}[thm]{Remark}

\usepackage{graphics,graphicx}

\newcommand{\sat}{\textsf{SAT}}
\newcommand{\ksat}{\textsf{k-SAT}}

\newcommand{\threesat}{\textsf{3-SAT}}
\newcommand{\threenae}{\textsf{3-NAE-SAT}}
\newcommand{\knae}{\textsf{k-NAE-SAT}}

\newcommand{\kcnf}{\textsf{k-CNF}}

\newcommand{\nae}{\textsf{NAE-SAT}}

\newcommand{\mathematica}{\textsf{Mathematica}}
\newcommand{\satlab}{\textsf{SATLab}}

\makeatother

\usepackage{babel}
\addto\extrasfrench{\providecommand{\fg}{\ifdim\lastskip>\z@\unskip\fi~\frqq}}

\begin{document}

\title{Second Moment Method on \ksat: a General Framework}

\author{Thomas Hugel and Yacine Boufkhad}

\maketitle
\begin{center}
LIAFA - Université Denis Diderot Paris 7 \& CNRS
\par
case 7014 - 75205 Paris Cedex 13
\end{center}

\begin{abstract}
We give a general framework implementing the Second Moment Method on k-SAT and discuss the conditions making the Second Moment Method work in this framework. As applications, we make the Second Moment Method work on boolean solutions and implicants. We extend this to the distributional model of k-SAT.
\end{abstract}

\section{Introduction to the Second Moment Method\index{Second Moment Method}\label{cha:SMM}}

\noindent Just as the First Moment Method is a way to upper-bound
the threshold of \ksat, so is the Second Moment Method a way to lower-bound
it. After a brief technical introduction to the Second Moment Method
in section \ref{sec:How-the-SMM}, we present in section \ref{sec:Survey-SMM}
a survey of the early attempts to lower-bound the threshold of \threesat{}
through the Second Moment Method. As in the First Moment Method, the
general idea is to count special kinds of solutions. However, the
selection of solutions is not the same as in the First Moment Method:
\begin{enumerate}
\item in the First Moment Method, we considered random variables $X$ such
that satisfiability implies $X\geq1$; setting $X$ to be the number
of solutions yields an upper bound of $5.191$ for the threshold of
\threesat; and we tried to select the least solutions;
\item in the Second Moment Method, we shall consider random variables $X$
such that unsatisfiability implies $X=0$ (see section \ref{sec:How-the-SMM});
setting $X$ to be the number of solutions yields a lower bound of
$0$ for the threshold of \threesat{} (see section \ref{sub:Second-Moment-of-Solutions});
here the criterion to select solutions is quite different: we are
interested in subsets of solutions having low variance.
\end{enumerate}
~

In a breakthrough paper, Achlioptas \& Peres - 2004 \cite{Achlioptas2004}
succeeded with the Second Moment Method on \ksat, establishing a
lower bound of $2.68$ for \threesat{} and an asymptotically tight
lower bound of $2^{k}\ln2-O\left(k\right)$ for \ksat. It turns out
that the currently best lower bound of the \threesat{} threshold
($3.52$) was obtained by another way: analyzing algorithms finding
solutions with high probability, cf.~Kaporis, Kirousis \& Lalas \cite{Kaporis2006}
and Hajiaghayi \& Sorkin \cite{Hajiaghayi2003}.

The purpose of the next chapters is to make the most of the Second
Moment Method on \ksat. To do so we take a different approach from
Achlioptas \& Peres'. In our framework we select solutions according
to the fraction of variables assigned $1$ and the fractions of the
different types of clauses (i.e.~the number of true literals occurrences
in the clauses). This framework is general enough to include boolean
solutions, implicants etc. However, tuning our parameters the best
way we could, we got numerical evidence that we could not obtain better
lower bounds than $2.83$.

The stumbling block we recurrently encountered is what we call the
independence point. It corresponds to the couples of independent solutions
in the subset of selected solutions. Even though solutions are independent,
the proportion of literals occurrences having a certain truth value
may not be independent between solutions. We got numerical evidence
that the Second Moment Method does not work if truth values of literals
occurrences are not independent. On the other hand, when they are
independent, we give a necessary condition for the Second Moment Method
to work, taking into account just the exponential equivalent of the
second moment at this point. This condition tells us that we must
select solutions having equal true and false surfaces (the surface
is just the total number of literals occurrences), which is very artificial
with respect to what we can observe with \satlab{} \cite{Hugel2010a}.
Using this condition, we could make the Second Moment Method work
numerically; however, since the lower bounds we get ($2.83$) are
far below the currently best lower bound ($3.52$), we do not give
a rigorous (and tedious) proof of our lower bounds (to do so, the
exponential equivalent would not be enough, and calculations would
become quite involved).

The very restrictive conditions we encountered to make the Second
Moment Method work may be due to some weaknesses of our framework.
We do not claim that the Second Moment Method is doomed to perpetual
failure on \ksat. We only hope to shed a small ray of light onto
it. This work is very fresh, still in progress, and has not been published.

In section \ref{cha:General-Framework-SMM} we present a general framework
implementing the Second Moment Method on random \ksat{} under various
settings (distributions of signs, implicants…). Section \ref{cha:Distributional-SM}
is dedicated to a variation of the framework presented in section
\ref{cha:General-Framework-SMM} for the needs of distributional models.

\subsection{How the Second Moment Method Is Supposed to Work\label{sec:How-the-SMM}}

Let us recall how the Second Moment Method is supposed to work: given
an event $A$, we want to show that $\mathrm{Pr}\left(A\right)$ tends
to $1$ but we don't have access to $\mathrm{Pr}\left(A\right)$.
Instead we use the first and the second moments of a non-negative
random variable $X$ such that $\mathrm{Pr}\left(A\right)\geq\mathrm{Pr}\left(X>0\right)$,
i.e. $X$ must be $0$ when $A$ does not hold. For our problem \threesat,
$A$ is the event {}``a formula is satisfiable''. The simplest choice
for $X$ is of course the number of solutions.
\begin{enumerate}
\item The first thing to notice is that if we show that $\mathrm{Pr}\left(A\right)$
is lower-bounded by a positive constant, then it tends to $1$. Why?
Because Friedgut \& Bourgain's theorem \cite{Friedgut1999} established
a sharp threshold for random \ksat;
\item In order to prove that $\mathrm{Pr}\left(X>0\right)$ is bounded away
from zero, we use the following classical identity: \begin{eqnarray}
\mathrm{Pr}\left(X>0\right) & \geq & \frac{\left(\mathrm{E}X\right)^{2}}{\mathrm{E}X^{2}}\enskip.\label{eq:SMCS}\end{eqnarray}
To prove it, use the fact that $\mathrm{E}X=\mathrm{E}\left(X\boldsymbol{1}_{X>0}\right)$
since $X\geq0$, and apply the Cauchy-Schwartz inequality to it: $\left(\mathrm{E}X\right)^{2}\leq\mathrm{E}X^{2}\mathrm{E}\boldsymbol{1}_{X>0}^{2}$.
So in particular note that\begin{eqnarray}
\frac{\left(\mathrm{E}X\right)^{2}}{\mathrm{E}X^{2}} & \leq & 1\enskip.\label{eq:SM>=PM2}\end{eqnarray}

\end{enumerate}
Just as the first moment is fairly easy to compute, so is the second
moment. Let $X$ be the number of assignments having some property
$\mathcal{P}$ (\emph{$\mathcal{P}$} might be {}``be a solution''
or {}``be a black and red polka-dot solution''):\begin{eqnarray*}
\mathrm{E}\left(X^{2}\right) & = & \mathrm{E}\left(\left(\sum_{\sigma\mathrm{\, assignment}}\boldsymbol{1}_{\sigma\in\mathcal{P}}\right)^{2}\right)\\
 & = & \mathrm{E}\left(\sum_{\sigma,\tau\mathrm{\, assignments}}\boldsymbol{1}_{\sigma\in\mathcal{P}}\boldsymbol{1}_{\tau\in\mathcal{P}}\right)\\
 & = & \sum_{\sigma,\tau\mathrm{\, assignments}}\mathrm{E}\boldsymbol{1}_{\sigma\in\mathcal{P}\land\tau\in\mathcal{P}}\\
 & = & \sum_{\sigma,\tau\mathrm{\, assignments}}\mathrm{Pr}\left(\sigma\in\mathcal{P}\land\tau\in\mathcal{P}\right)\enskip.\end{eqnarray*}

In the case of satisfiability however, we are going to see that the
Second Moment Method is much more difficult to implement than the
First Moment Method. The reason is that in general $\mathrm{E}X^{2}$
tends to be exponentially greater than $\left(\mathrm{E}X\right)^{2}$,
so equation \ref{eq:SMCS} just says that $\mathrm{Pr}\left(X>0\right)\geq0$,
which is not very informative… Thus the challenge is to find out a
set of solutions having low variance.

\subsection{Use of the Second Moment Method for Lower-Bounding the Threshold
of \texttt{\ksat\label{sec:Survey-SMM}}}

In this section we make a survey of different applications of the
Second Moment Method to lower-bound the threshold of random \ksat.
The model considered here is uniform random drawing of \ksat.

\subsubsection{Second Moment of Solutions\label{sub:Second-Moment-of-Solutions}}

Here is a perfect example of the failure of the Second Moment Method.
Namely the lower bound obtained by the Second Moment Method of solutions
is $c=0$. If $X$ is just the number of solutions, then $\mathrm{E}X=2^{n}\left(1-\frac{1}{2^{k}}\right)^{cn}$.
So let us now compute the second moment.

To do so we need an extra parameter $\mu$, representing the proportion
of variables changing values between two solutions.
\begin{enumerate}
\item total number of couples of assignments:

\begin{enumerate}
\item choose the value of variables assigned $0$ or $1$ in the first assignment:
$2^{n}$;
\item choose the subset of variables assigned different values in both assignments:
${n \choose \mu n}$;
\end{enumerate}
\item probability for a couple of assignments to be a couple of solutions:
as noted by Achlioptas \& Peres \cite{Achlioptas2004}, it is easier
to compute the probability that a clause breaks two given assignments,
since it is $\left(\frac{1-\mu}{2}\right)^{k}$. Using then the fact
that $\mathrm{Pr}\left(A\cap B\right)=1-\mathrm{Pr}\left(\overline{A}\cup\overline{B}\right)=1-\mathrm{Pr}\left(\overline{A}\right)-\mathrm{Pr}\left(\overline{B}\right)+\mathrm{Pr}\left(\overline{A}\cap\overline{B}\right)$,
it follows that the probability for a clause not to break any of both
assignments is $g\left(\mu\right)=1-\frac{2}{2^{k}}+\left(\frac{1-\mu}{2}\right)^{k}$.
\end{enumerate}
Thus the second moment is:\begin{eqnarray*}
\mathrm{E}X^{2} & = & 2^{n}\sum_{\substack{0\leq\mu\leq1\\
\mu n\in\boldsymbol{N}}
}{n \choose \mu n}g\left(\mu\right)^{cn}\enskip.\end{eqnarray*}

Let us look at the exponential equivalent of this quantity: \[
\mathrm{E}X^{2}\asymp\max_{0\leq\mu\leq1}\left(\frac{2}{\mu^{\mu}\left(1-\mu\right)^{\mu}}g\left(\mu\right)^{c}\right)^{n}\enskip.\]

As explained by Achlioptas \& Peres \cite{Achlioptas2004}, it turns
out that when $\mu=\frac{1}{2}$, this is precisely $\left(\mathrm{E}X\right)^{2}$.
Now the function $\mu\mapsto\frac{1}{\mu^{\mu}\left(1-\mu\right)^{1-\mu}}$
has its maximum at $\mu=\frac{1}{2}$, whereas $g\left(\mu\right)$
is strictly decreasing over $\left(0,1\right)$. Consequently, at
any positive ratio $c$, the maximum of $\frac{2}{\mu^{\mu}\left(1-\mu\right)^{\mu}}g\left(\mu\right)^{c}$
occurs at $\mu<\frac{1}{2}$ and $\mathrm{E}X^{2}$ is exponentially
greater than $\left(\mathrm{E}X\right)^{2}$. So $\frac{\left(\mathrm{E}X\right)^{2}}{\mathrm{E}X^{2}}$
tends to zero, and we only get that $\mathrm{Pr}\left(X>0\right)\geq0$…

\subsubsection{Balancing True and False Surfaces\label{sub:Balancing-True-and-False}}

Achlioptas \& Moore - 2002 \cite{Achlioptas2002} noticed that $g\left(\mu\right)$
is locally maximal at $\mu=\frac{1}{2}$ in \knae{} because in this
problem $g\left(\mu\right)=1-\frac{4}{2^{k}}+\frac{1}{2^{k}}\left(\left(1-\mu\right)^{k}+\mu^{k}\right)$
is symmetric in $\mu$. We recall that in \nae, when an assignment
is a solution, then the opposite assignment is a solution as well.
Thus this problem contains some symmetry. Using this remark, Achlioptas
\& Moore were able to establish a tight lower bound on the \knae{}
threshold. And since a \nae{} solution is a solution of standard
\sat, they got the following lower bound of the \ksat{} threshold:
$2^{k-1}\ln2-O\left(1\right)$.

Achlioptas \& Peres - 2004 \cite{Achlioptas2004} put some weights
onto the solutions of standard \sat{} and got a lower bound of $2^{k}\ln2-O\left(k\right)$.
(thus almost matching the asymptotic upper bound of $2^{k}\ln2$).
The weights they put are of the form $\lambda^{\mathrm{true\, surface}}\boldsymbol{1}_{\mathrm{SAT}}$,
where the true surface\index{surface} is the number of occurrences
of true literals under the solution. Assignments which are not solutions
must be discarded because the Second Moment Method requires to count
$0$ when there is no solution, as explained in section \ref{sec:How-the-SMM}.
In the particular case of \threesat{} they got a lower bound of $2.54$
(and even $2.68$ with a refinement).

We are going to implement the Second Moment Method without any weights;
so how shall we control the balance of true and false surfaces? Our
control parameters will be $\beta_{t}$, the fraction of clauses having
$t$ true literals. Then the true surface will be $\beta_{1}+2\beta_{2}+3\beta_{3}$
and the false surface will be $2\beta_{1}+\beta_{2}$. With this parameters
we are able to make the Second Moment work. However, we are not able
to achieve a better lower bound than $2.833$, see section \ref{sub:SM-Boolean-Solutions}.

 Moreover our approach is quite general and enables us to make the
Second Moment Method work on implicants as well.

\section{A General Framework for the Second Moment Method on \ksat\label{cha:General-Framework-SMM}}

Here we present a general framework for the Second Moment Method on
\ksat. Section \ref{sub:General-Framework-Preliminaries} introduces
all ingredients we need: values, signs, truth values, types of clauses
and surfaces. Then in section \ref{sec:General-SM:FM} we give the
expression of the first moment of the solutions under these settings;
the expression of the second moment is given in section \ref{sec:General-SM:SM}.
Bringing together the second moment and the constraints, we use the
Lagrange multipliers method in section \ref{sec:SM-Lagrangian}.

One point in the space of the variables is very important in the Second
Moment Method: this is what we call the independence point. It is
important because it makes $\frac{\mathrm{E}X^{2}}{\left(\mathrm{E}X\right)^{2}}=1$
(see conditions in theorem \ref{thm:1-at-independence} of section
\ref{sec:Independence-Point}). Thus if we want the Second Moment
to work, we must be careful that this point should be stationary.

We apply this general framework to boolean solutions (section \ref{sub:SM-Boolean-Solutions})
and to implicants (section \ref{sub:SM-Implicants}).

We discuss the relevance of the Second Moment Method for lower-bounding
the \ksat{} threshold in section \ref{sec:Confrontation-SM-vs-Reality},
where we use \satlab{} to confront the theoretical requirements we
obtained with reality.

\subsection{Preliminaries \label{sub:General-Framework-Preliminaries}}

\subsubsection{Values \index{value} \label{sub:SM-Values}}

First of all we have $n$ \emph{variables}. An \emph{assignment} gives
each variable a \emph{value} taken from a given \emph{domain} $D$:
\begin{itemize}
\item in the case of boolean satisfiability, $D=\left\{ 0,1\right\} $;
\item in the case of implicants, $D=\left\{ 0,1,*\right\} $.
\end{itemize}
Given an assignment, for all $a\in D$, we denote by $\delta_{a}$
the proportion of variables assigned value $a$: \begin{eqnarray}
\sum_{a\in D}\delta_{a} & = & 1\enskip.\label{eq:sum_deltas_1}\end{eqnarray}

Given two assignments $S_{1}$ and $S_{2}$, for all $\left(a,b\right)\in D^{2}$,
we denote by $\mu_{a,b}$ the proportion of variables assigned value
$a$ in $S_{1}$ and value $b$ in $S_{2}$: \begin{eqnarray}
\sum_{a\in D}\mu_{a,b} & = & \delta_{b}\enskip;\label{eq:delta_b}\\
\sum_{b\in D}\mu_{a,b} & = & \delta_{a}\enskip.\label{eq:delta_a}\end{eqnarray}

So we are going to consider some of the $\mu_{a,b}$'s as functions
of the other ones, assuming that equations \ref{eq:delta_b} and \ref{eq:delta_a}
are satisfied. We shall refer to the remaining $\mu_{a,b}$'s as a
generic variable $\mu$, cf. section \ref{sub:Derivative-mu}.

\subsubsection{Signs \index{sign}}

In a \kcnf{} formula with $cn$ \emph{clauses}, we have $kcn$ \emph{occurrences}
of variables, each having a \emph{sign} $s\in S$. In the case of
boolean satisfiability as well as in the case of implicants, $S=\left\{ +,-\right\} $;

For all $s\in S$, we denote by $\rho_{s}$ the proportion of occurrences
having sign $s$: \begin{eqnarray}
\sum_{s\in S}\rho_{s} & = & 1\enskip.\label{eq:sum_rhos_1}\end{eqnarray}

\subsubsection{Truth Values \index{truth value}}

The combination of a sign and a value yields a \emph{truth value}
$v\in\mathcal{V}$. Here is an example of a classical truth table
with $\mathcal{V}=\left\{ T,F,*\right\} $: \begin{tabular}{|c|c|c|}
\hline 
 & $+$ & $-$\tabularnewline
\hline 
$0$ & $F$ & $T$\tabularnewline
\hline 
$1$ & $T$ & $F$\tabularnewline
\hline 
$*$ & $*$ & $*$\tabularnewline
\hline
\end{tabular}.

We use the following notation, for $a\in D$, $s\in S$ and $v\in\mathcal{V}$:\begin{eqnarray*}
\chi_{a,s,v} & = & \begin{cases}
1 & \mbox{if value \ensuremath{a}\,\ and sign \ensuremath{s}\,\ yield the truth value\,\ensuremath{v}}\enskip;\\
0 & \mbox{otherwise}\enskip.\end{cases}\end{eqnarray*}

Of course one sign and one value yield exactly one truth value: \begin{eqnarray*}
\sum_{v\in\mathcal{V}}\chi_{a,s,v} & = & 1\enskip.\end{eqnarray*}

So we shall denote by $a\otimes s$ the unique $v$ such that $\chi_{a,s,v}=1$.

Given an assignment, for all $v\in\mathcal{V}$, we denote by $\eta_{v}$
the proportion of literals occurrences having the truth value $v$:
\begin{eqnarray*}
\eta_{v} & = & \sum_{\substack{a\in D\\
s\in S}
}\chi_{a,s,v}\delta_{a}\rho_{s}\enskip.\end{eqnarray*}

Given two assignments $S_{1}$ and $S_{2}$, for all $\left(v,w\right)\in\mathcal{V}^{2}$,
we denote by $\varepsilon_{v,w}$ the proportion of literals occurrences
having the truth value $v$ in $S_{1}$ and the truth value $w$ in
$S_{2}$: \begin{eqnarray*}
\varepsilon_{v,w} & = & \sum_{\substack{\left(a,b\right)\in D^{2}\\
s\in S}
}\chi_{a,s,v}\chi_{b,s,w}\mu_{a,b}\rho_{s}\enskip.\end{eqnarray*}

\subsubsection{Clauses Types \index{clause type}}

A\emph{ clause type }is an element of $\mathcal{V}^{k}$, e.g. TTF
or {*}TF. Some types of clauses will be forbidden, such as FFF and
FF{*}. We denote by $\mathcal{T}$ the set of allowed types of clauses.
$\mathcal{T}$ is a subset of $\mathcal{V}^{k}$.
\begin{itemize}
\item in the case of boolean solutions of \threesat: $\mathcal{V}=\left\{ T,F\right\} $
and FFF is forbidden;
\item in the case of implicants of \threesat: $\mathcal{V}=\left\{ T,F,*\right\} $
and the allowed types of clauses are those containing at least one
$T$.
\end{itemize}
Given an assignment, for all $t\in\mathcal{T}$, we denote by $\beta_{t}$
the proportion of clauses of type $t$ (which is zero for all forbidden
types of clauses): \begin{eqnarray}
\sum_{t\in\mathcal{T}}\beta_{t} & = & 1\enskip.\label{eq:sumbetas1}\end{eqnarray}

Given two assignments $S_{1}$ and $S_{2}$, for all $\left(t,u\right)\in\mathcal{T}^{2}$,
we denote by $\gamma_{t,u}$ the proportion of clauses of type $t$
in solution $S_{1}$ and of type $u$ in solution $S_{2}$:\begin{eqnarray}
\sum_{t\in\mathcal{T}}\gamma_{t,u} & = & \beta_{u}\enskip;\label{eq:beta_u}\\
\sum_{u\in\mathcal{T}}\gamma_{t,u} & = & \beta_{t}\enskip.\label{eq:beta_t}\end{eqnarray}

\subsubsection{Surfaces \index{surface}}

Given an assignment, the \emph{surface} occupied by a truth value
$v$ is obtained by summing all occurrences of $v$ in the different
types of clauses: $\Sigma_{v}=\sum_{t\in\mathcal{T}}\beta_{t}\sum_{i=1}^{k}\boldsymbol{1}_{t_{i}=v}$.

Given two assignments $S_{1}$ and $S_{2}$, the \emph{surface} occupied
by a couple of truth values $\left(v,w\right)$ is obtained by summing
all occurrences of $\left(v,w\right)$ in the different types of clauses:
$\Xi_{v,w}=\sum_{\left(t,u\right)\in\mathcal{T}^{2}}\gamma_{t,u}\sum_{i=1}^{k}\boldsymbol{1}_{t_{i}=v\land u_{i}=w}$.

Surfaces are normalized to $k$ because the $\beta_{t}$'s sum up
to $1$:
\begin{fact}
$\sum_{v\in\mathcal{V}}\Sigma_{v}=k$ and $\sum_{\left(v,w\right)\in\mathcal{V}^{2}}\Xi_{v,w}=k$.\label{fac:total-surface}\end{fact}
\begin{proof}
\begin{eqnarray*}
\sum_{v\in\mathcal{V}}\Sigma_{v} & = & \sum_{v\in\mathcal{V}}\sum_{t\in\mathcal{T}}\beta_{t}\sum_{i=1}^{k}\boldsymbol{1}_{t_{i}=v}\\
 & = & \sum_{t\in\mathcal{T}}\beta_{t}\sum_{i=1}^{k}\sum_{v\in\mathcal{V}}\boldsymbol{1}_{t_{i}=v}\\
 & = & \sum_{t\in\mathcal{T}}\beta_{t}\sum_{i=1}^{k}1\\
 & = & k\sum_{t\in\mathcal{T}}\beta_{t}\\
 & = & k\mbox{ by constraint \ref{eq:sumbetas1}.}\end{eqnarray*}

The same proof works for the other sum, using constraints \ref{eq:beta_u}
and \ref{eq:sumbetas1}.
\end{proof}

\subsubsection{Symmetry of Occurrences \index{symmetry of occurrences}\label{sub:Symmetry-of-Occurrences}}

We say that there is \emph{symmetry of occurrences} when for all permutation
$\sigma$ of $\mathcal{V}^{k}$, $\beta_{t}=\beta_{\sigma\left(t\right)}$
(it follows that $\mathcal{T}$ is closed by permutation).
\begin{fact}
Symmetry of occurrences implies that $\Sigma_{v}=k\sum_{t\in\mathcal{T}}\boldsymbol{1}_{t_{1}=v}\beta_{t}$.\label{fac:Symmetry-of-occurrences}\end{fact}
\begin{proof}
Let us call $\sigma_{i}$ the permutation of $\mathcal{V}^{k}$ swapping
$t_{1}$ and $t_{i}$.\begin{eqnarray*}
\Sigma_{v} & = & \sum_{t\in\mathcal{T}}\beta_{t}\sum_{i=1}^{k}\boldsymbol{1}_{t_{i}=v}\\
 & = & \sum_{i=1}^{k}\sum_{t\in\mathcal{T}}\boldsymbol{1}_{t_{i}=v}\beta_{t}\\
 & = & \sum_{i=1}^{k}\sum_{t\in\mathcal{T}}\boldsymbol{1}_{\sigma_{i}\left(t\right)_{1}=v}\beta_{\sigma_{i}\left(t\right)}\mbox{ by symmetry of occurrences}\\
 & = & \sum_{i=1}^{k}\sum_{\sigma_{i}\left(t\right)\in\mathcal{T}}\boldsymbol{1}_{\sigma_{i}\left(t\right)_{1}=v}\beta_{\sigma_{i}\left(t\right)}\mbox{ because \ensuremath{\mathcal{T}}\,\ is closed by permutation}\\
 & = & k\sum_{t\in\mathcal{T}}\boldsymbol{1}_{t_{1}=v}\beta_{t}\enskip.\end{eqnarray*}
\end{proof}
\begin{fact}
If $\gamma_{t,u}=\beta_{t}\beta_{u}$, symmetry of occurrences implies
that $k\Xi_{v,w}=\Sigma_{v}\Sigma_{w}$.\label{fac:Indep-of-surfaces}\end{fact}
\begin{proof}
Let us call $\sigma_{i}$ the permutation of $\mathcal{V}^{k}$ swapping
$t_{1}$ and $t_{i}$.\begin{eqnarray*}
k\Xi_{v,w} & = & k\sum_{\left(t,u\right)\in\mathcal{T}^{2}}\gamma_{t,u}\sum_{i=1}^{k}\boldsymbol{1}_{t_{i}=v\land u_{i}=w}\\
 & = & k\sum_{\left(t,u\right)\in\mathcal{T}^{2}}\beta_{t}\beta_{u}\sum_{i=1}^{k}\boldsymbol{1}_{t_{i}=v}\boldsymbol{1}_{u_{i}=w}\mbox{ by independence}\\
 & = & k\sum_{i=1}^{k}\sum_{\left(\sigma_{i}\left(t\right),\sigma_{i}\left(u\right)\right)\in\mathcal{T}^{2}}\beta_{\sigma_{i}\left(t\right)}\beta_{\sigma_{i}\left(u\right)}\boldsymbol{1}_{\sigma_{i}\left(t\right)_{1}=v}\boldsymbol{1}_{\sigma_{i}\left(u\right)_{1}=w}\mbox{ by symmetry of occurrences}\\
 & = & k^{2}\sum_{\left(t,u\right)\in\mathcal{T}^{2}}\beta_{t}\beta_{u}\boldsymbol{1}_{t_{1}=v}\boldsymbol{1}_{u_{1}=w}\mbox{ because \ensuremath{\mathcal{T}}\,\ is closed by permutation}\\
 & = & \Sigma_{v}\Sigma_{w}\mbox{ by fact \ref{fac:Symmetry-of-occurrences}}.\end{eqnarray*}

\end{proof}
Symmetry of occurrences is quite natural and will be assumed from
now on. Note that $\delta_{a}$'s and $\beta_{t}$'s are parameters
of the first moment, so they may be chosen without any restriction,
except that they must sum up to $1$. They are our control parameters:
we can tune them as we wish in order to take into account only some
solutions. However, when the set of solutions defined by $\delta_{a}$'s
and $\beta_{t}$'s is determined, we have to consider all possible
couples of solutions. So the variables $\mu_{a,b}$'s and $\gamma_{t,u}$'s
of the second moment may not be chosen, but result from a maximization
process, as investigated in section \ref{sec:SM-Lagrangian}.

\subsection{Expression of the First Moment \label{sec:General-SM:FM}}

The first moment of the number $X$ of solutions can be split up into
the following factors: total number of assignments and probability
for an assignment to be a solution.
\begin{enumerate}
\item total number of assignments: choose subsets of variables assigned
$a\in D$: ${n \choose \dots\left(\delta_{a}n\right)_{a\in D}\dots}$.
\item probability for an assignment to be a solution:

\begin{enumerate}
\item we give each clause an allowed type $t\in\mathcal{T}$: ${cn \choose \dots\left(\beta_{t}cn\right)_{t\in\mathcal{T}}\dots}$.
\item probability for clauses to be constructed (variables + signs) according
to their types: $\prod_{t\in\mathcal{T}}\left(\prod_{i=1}^{k}\eta_{t_{i}}\right)^{\beta_{t}cn}$.
\end{enumerate}
\end{enumerate}
We denote by $\mathcal{P}$ the set of all families of non-negative
numbers $\left(\left(\delta_{a}\right)_{a\in D},\left(\beta_{t}\right)_{t\in\mathcal{T}}\right)$
satisfying constraints \ref{eq:sum_deltas_1} and \ref{eq:sumbetas1}.
We denote by $\mathcal{I}\left(n\right)$ the intersection of $\mathcal{P}$
with the multiples of $\frac{1}{n}$; we get the following expression
of the first moment:

\begin{eqnarray*}
\mathrm{E}X & = & \sum_{\left(\left(\delta_{a}\right)_{a\in D},\left(\beta_{t}\right)_{t\in\mathcal{T}}\right)\in\mathcal{I}\left(n\right)}T_{1}\left(n\right)\end{eqnarray*}

where

\begin{eqnarray*}
T_{1}\left(n\right) & = & {n \choose \dots\left(\delta_{a}n\right)_{a\in D}\dots}{cn \choose \dots\left(\beta_{t}cn\right)_{t\in\mathcal{T}}\dots}\left(\prod_{t\in\mathcal{T}}\left(\prod_{i=1}^{k}\eta_{t_{i}}\right)^{\beta_{t}}\right)^{cn}\\
 & = & {n \choose \dots\left(\delta_{a}n\right)_{a\in D}\dots}{cn \choose \dots\left(\beta_{t}cn\right)_{t\in\mathcal{T}}\dots}\left(\prod_{t\in\mathcal{T}}\prod_{i=1}^{k}\prod_{v\in\mathcal{V}}\eta_{v}^{\boldsymbol{1}_{t_{i}=v}\beta_{t}}\right)^{cn}\\
 & = & {n \choose \dots\left(\delta_{a}n\right)_{a\in D}\dots}{cn \choose \dots\left(\beta_{t}cn\right)_{t\in\mathcal{T}}\dots}\left(\prod_{v\in\mathcal{V}}\eta_{v}^{\sum_{t\in\mathcal{T}}\sum_{i=1}^{k}\boldsymbol{1}_{t_{i}=v}\beta_{t}}\right)^{cn}\\
 & = & {n \choose \dots\left(\delta_{a}n\right)_{a\in D}\dots}{cn \choose \dots\left(\beta_{t}cn\right)_{t\in\mathcal{T}}\dots}\left(\prod_{v\in\mathcal{V}}\eta_{v}^{\Sigma_{v}}\right)^{cn}\enskip.\end{eqnarray*}

The exponential equivalent of $T_{1}\left(n\right)$ is $T_{1}^{n}$,
with the following two equivalent forms:

\begin{eqnarray*}
T_{1} & = & \frac{1}{\prod_{a\in D}\delta_{a}^{\delta_{a}}}\left(\prod_{t\in\mathcal{T}}\left(\frac{\prod_{i=1}^{k}\eta_{t_{i}}}{\beta_{t}}\right)^{\beta_{t}}\right)^{c}\\
 & = & \frac{1}{\prod_{a\in D}\delta_{a}^{\delta_{a}}}\left(\frac{\prod_{v\in\mathcal{V}}\eta_{v}^{\Sigma_{v}}}{\prod_{t\in\mathcal{T}}\beta_{t}^{\beta_{t}}}\right)^{c}\enskip.\end{eqnarray*}

\subsection{Expression of the Second Moment \label{sec:General-SM:SM}}

The second moment of $X$ can be split up into the following factors:
total number of couples of assignments and probability for a couple
of assignments to be a couple of solutions.
\begin{enumerate}
\item total number of couples of assignments: ${n \choose \dots\left(\mu_{a,b}n\right)_{\left(a,b\right)\in D^{2}}\dots}$;
\item probability for a couple of assignments to be a couple of solutions:

\begin{enumerate}
\item we give each clause an allowed type $t\in\mathcal{T}$ in solution
$S_{1}$ and another $u\in\mathcal{T}$ in solution $S_{2}$: ${cn \choose \dots\left(\gamma_{t,u}cn\right)_{\left(t,u\right)\in\mathcal{T}^{2}}\dots}$;
\item probability for clauses to be constructed (variables + signs) according
to their types: $\prod_{\left(t,u\right)\in\mathcal{T}^{2}}\left(\prod_{i=1}^{k}\varepsilon_{t_{i},u_{i}}\right)^{\gamma_{t,u}cn}$.
\end{enumerate}
\end{enumerate}
We denote by $\mathcal{P}_{2}$ the set of all families of non-negative
numbers $\left(\left(\mu_{a,b}\right)_{\left(a,b\right)\in D^{2}},\left(\gamma_{t,u}\right)_{\left(t,u\right)\in\mathcal{T}^{2}}\right)$
satisfying constraints \ref{eq:delta_b}, \ref{eq:delta_a}, \ref{eq:beta_u}
and \ref{eq:beta_t}. We denote by $\mathcal{I}_{2}\left(n\right)$
the intersection of $\mathcal{P}_{2}$ with the multiples of $\frac{1}{n}$;
we get the following expression of the second moment:

\begin{eqnarray*}
\mathrm{E}X^{2} & = & \sum_{\left(\left(\mu_{a,b}\right)_{\left(a,b\right)\in D^{2}},\left(\gamma_{t,u}\right)_{\left(t,u\right)\in\mathcal{T}^{2}}\right)\in\mathcal{I}_{2}\left(n\right)}T_{2}\left(n\right)\end{eqnarray*}

where

\begin{eqnarray*}
 &  & T_{2}\left(n\right)\\
 & = & {n \choose \dots\left(\mu_{a,b}n\right)_{\left(a,b\right)\in D^{2}}\dots}{cn \choose \dots\left(\gamma_{t,u}cn\right)_{\left(t,u\right)\in\mathcal{T}^{2}}\dots}\left(\prod_{\left(t,u\right)\in\mathcal{T}^{2}}\left(\prod_{i=1}^{k}\varepsilon_{t_{i},u_{i}}\right)^{\gamma_{t,u}}\right)^{cn}\\
 & = & {n \choose \dots\left(\mu_{a,b}n\right)_{\left(a,b\right)\in D^{2}}\dots}{cn \choose \dots\left(\gamma_{t,u}cn\right)_{\left(t,u\right)\in\mathcal{T}^{2}}\dots}\left(\prod_{\left(t,u\right)\in\mathcal{T}^{2}}\prod_{i=1}^{k}\prod_{\left(v,w\right)\in\mathcal{V}^{2}}\varepsilon_{v,w}^{\boldsymbol{1}_{t_{i}=v\land u_{i}=w}\gamma_{t,u}}\right)^{cn}\\
 & = & {n \choose \dots\left(\mu_{a,b}n\right)_{\left(a,b\right)\in D^{2}}\dots}{cn \choose \dots\left(\gamma_{t,u}cn\right)_{\left(t,u\right)\in\mathcal{T}^{2}}\dots}\left(\prod_{\left(v,w\right)\in\mathcal{V}^{2}}\varepsilon_{v,w}^{\sum_{\left(t,u\right)\in\mathcal{T}^{2}}\sum_{i=1}^{k}\boldsymbol{1}_{t_{i}=v\land u_{i}=w}\gamma_{t,u}}\right)^{cn}\\
 & = & {n \choose \dots\left(\mu_{a,b}n\right)_{\left(a,b\right)\in D^{2}}\dots}{cn \choose \dots\left(\gamma_{t,u}cn\right)_{\left(t,u\right)\in\mathcal{T}^{2}}\dots}\left(\prod_{\left(v,w\right)\in\mathcal{V}^{2}}\varepsilon_{v,w}^{\Xi_{v,w}}\right)^{cn}\enskip.\end{eqnarray*}

The exponential equivalent of $T_{2}\left(n\right)$ is $T_{2}^{n}$
with the following two equivalent forms:

\begin{eqnarray*}
T_{2} & = & \frac{1}{\prod_{\left(a,b\right)\in D^{2}}\mu_{a,b}^{\mu_{a,b}}}\left(\prod_{\left(t,u\right)\in\mathcal{T}^{2}}\left(\frac{\prod_{i=1}^{k}\varepsilon_{t_{i},u_{i}}}{\gamma_{t,u}}\right)^{\gamma_{t,u}}\right)^{c}\\
 & = & \frac{1}{\prod_{\left(a,b\right)\in D^{2}}\mu_{a,b}^{\mu_{a,b}}}\left(\frac{\prod_{\left(v,w\right)\in\mathcal{V}^{2}}\varepsilon_{v,w}^{\Xi_{v,w}}}{\prod_{\left(t,u\right)\in\mathcal{T}^{2}}\gamma_{t,u}^{\gamma_{t,u}}}\right)^{c}\enskip.\end{eqnarray*}

\subsection{Expression of the Lagrangian\label{sec:SM-Lagrangian}}

When the parameters of the first moment (i.e. $\left(\delta_{a}\right)_{a\in D},\left(\beta_{t}\right)_{t\in\mathcal{T}}$)
are chosen, $T_{2}$ must be maximized under constraints \ref{eq:delta_b},
\ref{eq:delta_a}, \ref{eq:beta_u} and \ref{eq:beta_t}. That leads
us to use the Lagrange multipliers method.

As explained in section \ref{sub:SM-Values}, we are going to consider
some of the $\mu_{a,b}$'s as functions of the other ones, assuming
that equations \ref{eq:delta_b} and \ref{eq:delta_a} are satisfied.
We shall refer to the remaining $\mu_{a,b}$'s as a generic variable
$\mu$. So we define the following Lagrangian:

\begin{eqnarray*}
\Lambda & = & -\sum_{\left(a,b\right)\in D^{2}}\mu_{a,b}\ln\frac{\mu_{a,b}}{e}-c\sum_{\left(t,u\right)\in\mathcal{T}^{2}}\gamma_{t,u}\ln\frac{\gamma_{t,u}}{e}+c\sum_{\left(v,w\right)\in\mathcal{V}^{2}}\Xi_{v,w}\ln\varepsilon_{v,w}\\
 &  & +c\sum_{t\in\mathcal{T}}\left(\ln f_{t}\right)\left(\sum_{u\in\mathcal{T}}\gamma_{t,u}-\beta_{t}\right)+c\sum_{u\in\mathcal{T}}\left(\ln g_{u}\right)\left(\sum_{t\in\mathcal{T}}\gamma_{t,u}-\beta_{u}\right)\enskip.\end{eqnarray*}

\subsubsection{Derivative with respect to $\gamma_{t,u}$}

\begin{eqnarray*}
\frac{\partial\Lambda}{\partial\gamma_{t,u}} & = & -c\ln\gamma_{t,u}+c\sum_{i=1}^{k}\ln\varepsilon_{t_{i},u_{i}}+c\ln f_{t}+c\ln g_{u}\enskip.\end{eqnarray*}

Canceling out this derivative yields:

\begin{eqnarray}
\gamma_{t,u} & = & f_{t}g_{u}\prod_{i=1}^{k}\varepsilon_{t_{i},u_{i}}\enskip.\label{eq:gamma_tu}\end{eqnarray}

\subsubsection{Derivative with respect to $\mu$\label{sub:Derivative-mu}}

\begin{eqnarray}
\frac{\partial\Lambda}{\partial\mu} & = & \sum_{\left(a,b\right)\in D^{2}}\frac{\partial\Lambda}{\partial\mu_{a,b}}\frac{\partial\mu_{a,b}}{\partial\mu}\nonumber \\
 & = & \sum_{\left(a,b\right)\in D^{2}}\frac{\partial\mu_{a,b}}{\partial\mu}\left(-\ln\mu_{a,b}+c\sum_{\left(v,w\right)\in\mathcal{V}^{2}}\frac{\partial\varepsilon_{v,w}}{\partial\mu_{a,b}}\frac{\Xi_{v,w}}{\varepsilon_{v,w}}\right)\nonumber \\
 & = & -\sum_{\left(a,b\right)\in D^{2}}\frac{\partial\mu_{a,b}}{\partial\mu}\ln\mu_{a,b}+c\sum_{\left(a,b\right)\in D^{2}}\frac{\partial\mu_{a,b}}{\partial\mu}\sum_{\left(v,w\right)\in\mathcal{V}^{2}}\frac{\Xi_{v,w}}{\varepsilon_{v,w}}\sum_{s\in S}\chi_{a,s,v}\chi_{b,s,w}\rho_{s}\nonumber \\
 & = & -\sum_{\left(a,b\right)\in D^{2}}\frac{\partial\mu_{a,b}}{\partial\mu}\ln\mu_{a,b}+c\sum_{\left(a,b\right)\in D^{2}}\frac{\partial\mu_{a,b}}{\partial\mu}\sum_{s\in S}\rho_{s}\frac{\Xi_{a\otimes s,b\otimes s}}{\varepsilon_{a\otimes s,b\otimes s}}\enskip.\label{eq:der_mu}\end{eqnarray}

Canceling out this derivative is somewhat tricky in general, so we
are going to focus on some simple particular cases and otherwise end
up calculations numerically with \mathematica…

So let us consider first a particularly simplifying case, i.e. when
$\varepsilon_{v,w}=\eta_{v}\eta_{w}$.

\subsection{Independence Point - Discussion about $\varepsilon_{v,w}$ \label{sec:Independence-Point}}

We define the independence point \index{independence point} in the
polytope $\mathcal{P}_{2}$ to be the point where $\mu_{a,b}=\delta_{a}\delta_{b}$
and $\gamma_{t,u}=\beta_{t}\beta_{u}$. This point is of major interest
because it make $\frac{T_{2}}{T_{1}^{2}}=1$ if $\varepsilon_{v,w}=\eta_{v}\eta_{w}$
(see theorem \ref{thm:1-at-independence}). More surprisingly, it
turns out that there seems to be a dichotomy in the success / failure
of the Second Moment Method, regarding $\varepsilon_{v,w}$ at the
independence point:
\begin{itemize}
\item if $\varepsilon_{v,w}=\eta_{v}\eta_{w}$ then we are able to find
a necessary and sufficient condition on the first moment parameters
for the Second Moment Method to give a non trivial lower bound in
all models we considered:

\begin{itemize}
\item boolean solutions, see section \ref{sub:SM-Boolean-Solutions};
\item implicants, see section \ref{sub:SM-Implicants};
\item distributional model, see chapter \ref{cha:Distributional-SM};
\end{itemize}
\item otherwise numerical calculations give us empirical evidence that the
Second Moment Method fails to give any non trivial lower bound to
the threshold; indeed even when we have almost this identity, the
ratio $\frac{T_{2}}{T_{1}^{2}}$ is strictly greater than $1$ for
any positive ratio $c$ (see figures \ref{fig:delta1/2-eq} and \ref{fig:delta1/2-prop}
at the end of section \ref{sub:SM-Boolean-Solutions} and figure \ref{fig:delta_tp-omega=00003D1}
in chapter \ref{cha:Distributional-SM}).\end{itemize}
\begin{conjecture}
The Second Moment Method works only if $\varepsilon_{v,w}=\eta_{v}\eta_{w}$.\label{con:The-Second-Moment}
\end{conjecture}
This conjecture echoes the following theorem.
\begin{thm}
At the independence point (i.e. $\mu_{a,b}=\delta_{a}\delta_{b}$
and $\gamma_{t,u}=\beta_{t}\beta_{u}$), if $\varepsilon_{v,w}=\eta_{v}\eta_{w}$,
then $\frac{T_{2}}{T_{1}^{2}}=1$.\label{thm:1-at-independence}\end{thm}
\begin{proof}
Let us recall that

\begin{eqnarray*}
\frac{T_{2}}{T_{1}^{2}} & = & \frac{\left(\prod_{a\in D}\delta_{a}^{\delta_{a}}\right)^{2}}{\prod_{\left(a,b\right)\in D^{2}}\mu_{a,b}^{\mu_{a,b}}}\left(\left(\prod_{t\in\mathcal{T}}\left(\frac{\beta_{t}}{\prod_{i=1}^{k}\eta_{t_{i}}}\right)^{\beta_{t}}\right)^{2}\prod_{\left(t,u\right)\in\mathcal{T}^{2}}\left(\frac{\prod_{i=1}^{k}\varepsilon_{t_{i},u_{i}}}{\gamma_{t,u}}\right)^{\gamma_{t,u}}\right)^{c}\enskip.\end{eqnarray*}

So at the independence point:

\begin{eqnarray*}
\frac{T_{2}}{T_{1}^{2}} & = & \frac{\left(\prod_{a\in D}\delta_{a}^{\delta_{a}}\right)\left(\prod_{b\in D}\delta_{b}^{\delta_{b}}\right)}{\prod_{\left(a,b\right)\in D^{2}}\left(\delta_{a}\delta_{b}\right)^{\delta_{a}\delta_{b}}}\\
 &  & \cdot\left(\frac{\left(\prod_{t\in\mathcal{T}}\beta_{t}^{\beta_{t}}\right)\left(\prod_{u\in\mathcal{T}}\beta_{u}^{\beta_{u}}\right)}{\prod_{\left(t,u\right)\in\mathcal{T}^{2}}\left(\beta_{t}\beta_{u}\right)^{\beta_{t}\beta_{u}}}\cdot\frac{\prod_{\left(t,u\right)\in\mathcal{T}^{2}}\left(\left(\prod_{i=1}^{k}\eta_{t_{i}}\right)\left(\prod_{i=1}^{k}\eta_{u_{i}}\right)\right)^{\beta_{t}\beta_{u}}}{\left(\prod_{t\in\mathcal{T}}\left(\prod_{i=1}^{k}\eta_{t_{i}}\right)^{\beta_{t}}\right)\left(\prod_{u\in\mathcal{T}}\left(\prod_{i=1}^{k}\eta_{u_{i}}\right)^{\beta_{u}}\right)}\right)^{c}\enskip.\end{eqnarray*}

To show that this quantity is indeed $1$, we shall use the following
fact:\end{proof}
\begin{fact}
\label{fac:factorization}Let \emph{$\mathcal{X}$} and $\mathcal{Y}$
be two finite sets. We assume that $\sum_{x\in\mathcal{X}}e_{x}=\sum_{y\in\mathcal{Y}}f_{y}=1$.\\
 Then $\prod_{\left(x,y\right)\in\mathcal{X}\times\mathcal{Y}}\left(b_{x}c_{y}\right)^{e_{x}f_{y}}=\left(\prod_{x\in\mathcal{X}}b_{x}^{e_{x}}\right)\left(\prod_{y\in\mathcal{Y}}c_{y}^{f_{y}}\right)$.\end{fact}
\begin{proof}
\begin{eqnarray*}
\prod_{\left(x,y\right)\in\mathcal{X}\times\mathcal{Y}}\left(b_{x}c_{y}\right)^{e_{x}f_{y}} & = & \prod_{\left(x,y\right)\in\mathcal{X}\times\mathcal{Y}}b_{x}^{e_{x}f_{y}}\prod_{\left(x,y\right)\in\mathcal{X}\times\mathcal{Y}}c_{y}^{e_{x}f_{y}}\\
 & = & \prod_{x\in\mathcal{X}}b_{x}^{e_{x}\sum_{y\in\mathcal{Y}}f_{y}}\prod_{y\in\mathcal{Y}}c_{y}^{f_{y}\sum_{x\in\mathcal{X}}e_{x}}\\
 & = & \prod_{x\in\mathcal{X}}b_{x}^{e_{x}}\prod_{y\in\mathcal{Y}}c_{y}^{f_{y}}\enskip.\end{eqnarray*}

Now to prove that the previous ratio $\frac{T_{2}}{T_{1}^{2}}$ is
$1$, it suffices to apply fact \ref{fac:factorization} 3 times:

with $\mathcal{X}=\mathcal{Y}=D$, $b_{x}=e_{x}=\delta_{a}$ and $c_{y}=f_{y}=\delta_{b}$,
which is possible thanks to equation \ref{eq:sum_deltas_1};

with $\mathcal{X}=\mathcal{Y}=\mathcal{T}$, $b_{x}=e_{x}=\beta_{t}$
and $c_{y}=f_{y}=\beta_{u}$, which is possible thanks to equation
\ref{eq:sumbetas1};

with $\mathcal{X}=\mathcal{Y}=\mathcal{T}$, $b_{x}=\prod_{i=1}^{k}\eta_{t_{i}}$,
$e_{x}=\beta_{t}$, $c_{y}=\prod_{i=1}^{k}\eta_{u_{i}}$ and $f_{y}=\beta_{u}$,
which is possible again thanks to equation \ref{eq:sumbetas1}. 
\end{proof}
Moreover, it turns out that the independence point satisfies constraints
\ref{eq:delta_b}, \ref{eq:delta_a}, \ref{eq:beta_u} and \ref{eq:beta_t},
thus $T_{2}$ must be stationary at the independence point if we want
the Second Moment Method to work (because $\frac{T_{2}}{T_{1}^{2}}$
must not exceed $1$ if we want to avoid the pitfall we encountered
in section \ref{sub:Second-Moment-of-Solutions}). Thus we have the
following necessary condition to make the Second Moment Method work:

\begin{eqnarray}
\frac{\partial\Lambda}{\partial\mu} & = & 0\mbox{ at the independence point.}\label{eq:cancel_out_deriv_mu}\end{eqnarray}

\begin{rem}
We show rigorously only the fact that the above conditions ($\varepsilon_{v,w}=\eta_{v}\eta_{w}$
and equation \ref{eq:cancel_out_deriv_mu}) are necessary to make
the Second Moment Method work, but not that they are sufficient. This
would require to handle the polynomial residues of the multinomials,
and the complete expressions of $\mathrm{E}X^{2}$ and $\left(\mathrm{E}X\right)^{2}$.
Since we show only negative results (i.e. bad lower bounds), this
tricky part is omitted.
\end{rem}

\subsection{Applications \label{sec:General-SM:Applications}}

\subsubsection{Boolean Solutions\label{sub:SM-Boolean-Solutions}}

\paragraph{Preliminaries.}
\begin{enumerate}
\item the domain of values is $D=\left\{ 0,1\right\} $; given a solution,
we call $\delta$ the fraction of variables assigned $1$; constraints
\ref{eq:delta_b} and \ref{eq:delta_a} become: \begin{eqnarray*}
\delta & = & \mu_{1,1}+\mu_{1,0}\\
\delta & = & \mu_{1,1}+\mu_{0,1}\\
1-\delta & = & \mu_{0,0}+\mu_{1,0}\\
1-\delta & = & \mu_{0,0}+\mu_{0,1}\end{eqnarray*}
so, if we define $\mu=\mu_{0,1}$, we have\foreignlanguage{french}{\begin{eqnarray*}
\mu_{1,1} & = & \delta-\mu\\
\mu_{0,0} & = & 1-\delta-\mu\\
\mu_{1,0} & = & \mu\\
\mu_{0,1} & = & \mu\end{eqnarray*}
}
\item the set of signs is $S=\left\{ +,-\right\} $; we call $\rho$ the
fraction of positive occurrences;
\item the set of truth values is $\mathcal{V}=\left\{ T,F\right\} $ and
the truth table is: \begin{tabular}{|c|c|c|}
\hline 
 & $+$ & $-$\tabularnewline
\hline 
$0$ & $F$ & $T$\tabularnewline
\hline 
$1$ & $T$ & $F$\tabularnewline
\hline
\end{tabular}; so\begin{eqnarray}
\eta_{T} & = & \rho\delta+\left(1-\rho\right)\left(1-\delta\right)\nonumber \\
\eta_{F} & = & \left(1-\rho\right)\delta+\rho\left(1-\delta\right)\nonumber \\
\varepsilon_{T,F}=\varepsilon_{F,T} & = & \mu\label{eq:sol-eps-TF-FT}\\
\varepsilon_{T,T} & = & \eta_{T}-\mu\label{eq:sol-eps-TT}\\
\varepsilon_{F,F} & = & \eta_{F}-\mu\label{eq:sol-eps-FF}\end{eqnarray}

\item the set of allowed types of clauses is $\mathcal{T}=D^{k}\backslash\left\{ F^{k}\right\} $.
\end{enumerate}

\paragraph{Condition for $\varepsilon_{v,w}=\eta_{v}\eta_{w}$ at the Independence
Point.}

The first thing to notice is that if $\varepsilon_{T,F}=\eta_{T}\eta_{F}$,
then the three other identities follow, because $\varepsilon_{T,F}+\varepsilon_{T,T}=\eta_{T}$
etc.

At the independence point, we have $\mu=\delta\left(1-\delta\right)$.
Thus \begin{eqnarray*}
\varepsilon_{T,F}-\eta_{T}\eta_{F} & = & \delta\left(1-\delta\right)-\left(\rho\delta+\left(1-\rho\right)\left(1-\delta\right)\right)\left(\left(1-\rho\right)\delta+\rho\left(1-\delta\right)\right)\\
 & = & \delta\left(1-\delta\right)-\rho\left(1-\rho\right)\left(\delta^{2}+\left(1-\delta\right)^{2}\right)-\left(\rho^{2}+\left(1-\rho\right)^{2}\right)\delta\left(1-\delta\right)\\
 & = & 2\rho\left(1-\rho\right)\delta\left(1-\delta\right)-\rho\left(1-\rho\right)\left(\delta^{2}+\left(1-\delta\right)^{2}\right)\\
 & = & -\rho\left(1-\rho\right)\left(2\delta-1\right)^{2}\enskip.\end{eqnarray*}

Consequently, $\varepsilon_{T,F}\leq\eta_{T}\eta_{F}$, with equality
iff $\rho\in\left\{ 0,1\right\} $ or $\delta=\frac{1}{2}$.

We discard the particular case of $\rho\in\left\{ 0,1\right\} $ (which
corresponds to monotone \ksat, always trivially satisfiable). It
turns out that as soon as $\delta\neq\frac{1}{2}$, we could not make
the Second Moment Method work: our numerical attempts revealed that
the ratio $\frac{T_{2}}{T_{1}^{2}}$ is strictly greater than $1$
for any positive ratio $c$.

On the other hand, when $\delta=\frac{1}{2}$ we could make the Second
Moment Method work, as follows.

\paragraph{Condition for the Second Moment Method to Work at $\delta=\frac{1}{2}$.}

As mentioned in section \ref{sec:Independence-Point}, stationarity
of the independence point implies that $\frac{\partial\Lambda}{\partial\mu}=0$.

Using equation \ref{eq:der_mu}:

\begin{eqnarray*}
\frac{\partial\Lambda}{\partial\mu} & = & -\sum_{\left(a,b\right)\in D^{2}}\frac{\partial\mu_{a,b}}{\partial\mu}\ln\mu_{a,b}+c\sum_{\left(a,b\right)\in D^{2}}\frac{\partial\mu_{a,b}}{\partial\mu}\sum_{s\in S}\rho_{s}\frac{\Xi_{a\otimes s,b\otimes s}}{\varepsilon_{a\otimes s,b\otimes s}}\\
 & = & \ln\frac{\mu_{1,1}\mu_{0,0}}{\mu_{1,0}\mu_{0,1}}+c\left(\left(\rho+\left(1-\rho\right)\right)\left(\frac{\Xi_{T,F}}{\varepsilon_{T,F}}+\frac{\Xi_{F,T}}{\varepsilon_{F,T}}-\frac{\Xi_{T,T}}{\varepsilon_{T,T}}-\frac{\Xi_{F,F}}{\varepsilon_{F,F}}\right)\right)\\
 & = & \ln\frac{\mu_{1,1}\mu_{0,0}}{\mu_{1,0}\mu_{0,1}}+c\left(\frac{\Xi_{T,F}}{\varepsilon_{T,F}}+\frac{\Xi_{F,T}}{\varepsilon_{F,T}}-\frac{\Xi_{T,T}}{\varepsilon_{T,T}}-\frac{\Xi_{F,F}}{\varepsilon_{F,F}}\right)\enskip.\end{eqnarray*}

At independence $\frac{\mu_{1,1}\mu_{0,0}}{\mu_{1,0}\mu_{0,1}}=\frac{\delta_{1}^{2}\delta_{0}^{2}}{\delta_{1}\delta_{0}\delta_{0}\delta_{1}}=1$;
moreover, assuming symmetry of occurrences, we may use independence
of surfaces (cf. fact \ref{fac:Indep-of-surfaces}):\begin{eqnarray*}
\frac{\partial\Lambda}{\partial\mu} & = & c\left(-\frac{\Sigma_{F}^{2}}{\eta_{F}^{2}}-\frac{\Sigma_{T}^{2}}{\eta_{T}^{2}}+2\frac{\Sigma_{F}\Sigma_{T}}{\eta_{F}\eta_{T}}\right)\\
 & = & c\left(\frac{\Sigma_{T}}{\eta_{T}}-\frac{\Sigma_{F}}{\eta_{F}}\right)^{2}\enskip.\end{eqnarray*}

Canceling out this derivative yields:\begin{eqnarray*}
\frac{\Sigma_{T}}{\eta_{T}} & = & \frac{\Sigma_{F}}{\eta_{F}}\enskip.\end{eqnarray*}

Since we assume $\delta=\frac{1}{2}$, we have $\eta_{T}=\eta_{F}=\frac{1}{2}$,
thus $\Sigma_{T}=\Sigma_{F}=\frac{k}{2}$.

It turns out that this condition is sufficient to make the Second
Moment Method work, and numerically we found a critical ratio $c=2.833$
for $\beta_{TFF}=\beta_{FTF}=\beta_{FFT}=0.197633$, $\beta_{TTF}=\beta_{TFT}=\beta_{FTT}=0.104733$
and $\beta_{TTT}=0.0929$. It is noticeable that this critical ratio
is the same for any value of $\rho$; this comes from the fact that
laying down $\eta_{T}=\eta_{F}=\frac{1}{2}$, equations \ref{eq:sol-eps-TF-FT},
\ref{eq:sol-eps-TT} and \ref{eq:sol-eps-FF} imply that equation
\ref{eq:Distributional-gamma_tu} has no dependence in $\rho$.

The First Moment Method applied with these settings (i.e. $\delta=\frac{1}{2}$)
yields a critical ratio of $3.783$ when $\beta_{TFF}=0.191$, which
means that such balanced solutions disappear far below the conjectured
threshold ratio of $4.25$. Thus we would like to evade the $\delta=\frac{1}{2}$
condition.

Moreover \satlab{} enables us to see that real solutions do not have
$\delta=\frac{1}{2}$, see our discussion in section \ref{sub:Non-Independence-Distances}.

\paragraph{Attempts to evade the $\delta=\frac{1}{2}$ condition.}

We plot $\ln\ln\frac{F_{2}}{F_{1}^{2}}$ for different values of $\delta$
and $\rho$, at a point satisfying our constraints \ref{eq:beta_u}
and \ref{eq:beta_t}. The expected value is $-\infty$ iff $\frac{F_{2}}{F_{1}^{2}}=1$.
We set the ratio $c=0.1$ (to be compared with $2.833$, where $\delta=\frac{1}{2}$
works).
\begin{enumerate}
\item setting $\Sigma_{T}=\Sigma_{F}=\frac{k}{2}$: only $\delta=\frac{1}{2}$
seems to make $\frac{F_{2}}{F_{1}^{2}}=1$, cf. figure \ref{fig:delta1/2-eq};%
\begin{figure}
\begin{centering}
\includegraphics[width=0.64\columnwidth]{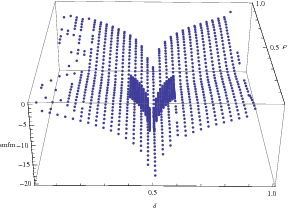}
\par\end{centering}

\caption{\label{fig:delta1/2-eq}$\Sigma_{T}=\Sigma_{F}=\frac{k}{2}$: $-\infty$
is obtained only when $\delta=\frac{1}{2}$.}

\end{figure}

\item setting $\Sigma_{T}=k\eta_{T}$ and $\Sigma_{F}=k\eta_{F}$: once
more, only $\delta=\frac{1}{2}$ seems to make $\frac{F_{2}}{F_{1}^{2}}=1$,
cf. figure \ref{fig:delta1/2-prop}.%
\begin{figure}
\begin{centering}
\includegraphics[width=0.64\columnwidth]{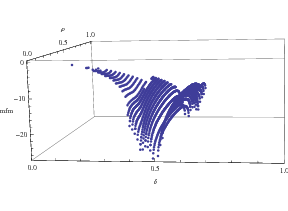}
\par\end{centering}

\caption{\label{fig:delta1/2-prop}$\Sigma_{T}=k\eta_{T}$ and $\Sigma_{F}=k\eta_{F}$:
$-\infty$ is obtained only when $\delta=\frac{1}{2}$.}

\end{figure}

\end{enumerate}
\clearpage{}

\subsubsection{Implicants \label{sub:SM-Implicants}}

An \emph{implicant}\index{implicant} is a partial assignment such
that every assignment of the non-assigned variables will yield a solution.
We represent the non-assigned value of variables by a $*$. We performed
the calculations on implicants with the hope that their variance might
be lower than the solutions'.

\paragraph{Preliminaries.}
\begin{enumerate}
\item the domain of values is $D=\left\{ 0,1,*\right\} $; given a solution,
we call $\delta$ the fraction of variables assigned $1$ and $\alpha$
the fraction of variables assigned $*$; constraints \ref{eq:delta_b}
and \ref{eq:delta_a} become: \begin{eqnarray*}
\delta & = & \mu_{1,1}+\mu_{1,0}+\mu_{1,*}\\
\delta & = & \mu_{1,1}+\mu_{0,1}+\mu_{*,1}\\
\alpha & = & \mu_{*,1}+\mu_{*,0}+\mu_{*,*}\\
\alpha & = & \mu_{1,*}+\mu_{0,*}+\mu_{*,*}\\
1-\delta-\alpha & = & \mu_{0,0}+\mu_{0,1}+\mu_{0,*}\\
1-\delta-\alpha & = & \mu_{0,0}+\mu_{1,0}+\mu_{*,0}\enskip,\end{eqnarray*}
so, if we define $\mu=\mu_{*,*}$, $\nu=\mu_{1,1}$, $\pi=\mu_{1,*}$
and $\pi^{'}=\mu_{*,1}$, we have\foreignlanguage{french}{\begin{eqnarray*}
\mu_{1,0} & = & \delta-\nu-\pi\\
\mu_{0,1} & = & \delta-\nu-\pi^{'}\\
\mu_{*,0} & = & \alpha-\mu-\pi^{'}\\
\mu_{0,*} & = & \alpha-\mu-\pi\\
\mu_{0,0} & = & 1-2\delta-2\alpha+\mu+\nu+\pi+\pi^{'}\enskip;\end{eqnarray*}
}
\item the set of signs is $S=\left\{ +,-\right\} $; we call $\rho$ the
fraction of positive occurrences;
\item the set of truth values is $\mathcal{V}=\left\{ T,F,*\right\} $ and
the truth table is: \begin{tabular}{|c|c|c|}
\hline 
 & $+$ & $-$\tabularnewline
\hline 
$0$ & $F$ & $T$\tabularnewline
\hline 
$1$ & $T$ & $F$\tabularnewline
\hline 
$*$ & $*$ & $*$\tabularnewline
\hline
\end{tabular}; so\begin{eqnarray*}
\eta_{T} & = & \rho\delta+\left(1-\rho\right)\left(1-\delta-\alpha\right)\\
\eta_{F} & = & \left(1-\rho\right)\delta+\rho\left(1-\delta-\alpha\right)\\
\eta_{*} & = & \alpha\end{eqnarray*}
and\begin{eqnarray*}
\varepsilon_{T,T} & = & \rho\nu+\left(1-\rho\right)\left(1-2\delta-2\alpha+\mu+\nu+\pi+\pi^{'}\right)\\
\varepsilon_{T,F} & = & \rho\left(\delta-\nu-\pi\right)+\left(1-\rho\right)\left(\delta-\nu-\pi^{'}\right)\\
\varepsilon_{F,T} & = & \rho\left(\delta-\nu-\pi^{'}\right)+\left(1-\rho\right)\left(\delta-\nu-\pi\right)\\
\varepsilon_{F,F} & = & \rho\left(1-2\delta-2\alpha+\mu+\nu+\pi+\pi^{'}\right)+\left(1-\rho\right)\nu\\
\varepsilon_{*,*} & = & \mu\\
\varepsilon_{T,*} & = & \rho\pi+\left(1-\rho\right)\left(\alpha-\mu-\pi\right)\\
\varepsilon_{*,T} & = & \rho\pi^{'}+\left(1-\rho\right)\left(\alpha-\mu-\pi^{'}\right)\\
\varepsilon_{F,*} & = & \rho\left(\alpha-\mu-\pi\right)+\left(1-\rho\right)\pi\\
\varepsilon_{*,F} & = & \rho\left(\alpha-\mu-\pi^{'}\right)+\left(1-\rho\right)\pi^{'}\enskip;\end{eqnarray*}

\item a clause type is allowed iff it contains at least one $T$.
\end{enumerate}

\paragraph{Condition for $\varepsilon_{v,w}=\eta_{v}\eta_{w}$ at the Independence
Point.}

The independence point is defined by $\mu=\alpha^{2}$, $\nu=\delta^{2}$
and $\pi=\pi^{'}=\alpha\delta$. So at this point we have:

\begin{eqnarray*}
\varepsilon_{T,T} & = & \rho\delta^{2}+\left(1-\rho\right)\left(1-\delta-\alpha\right)^{2}\\
\varepsilon_{T,F}=\varepsilon_{F,T} & = & \delta\left(1-\delta-\alpha\right)\\
\varepsilon_{F,F} & = & \rho\left(1-\delta-\alpha\right)^{2}+\left(1-\rho\right)\delta^{2}\\
\varepsilon_{*,*} & = & \alpha^{2}=\eta_{*}^{2}\\
\varepsilon_{T,*}=\varepsilon_{*,T} & = & \alpha\left(\rho\delta+\left(1-\rho\right)\left(1-\delta-\alpha\right)\right)=\eta_{T}\eta_{*}\\
\varepsilon_{F,*}=\varepsilon_{*,F} & = & \alpha\left(\rho\left(1-\delta-\alpha\right)+\left(1-\rho\right)\delta\right)=\eta_{F}\eta_{*}\enskip.\end{eqnarray*}

The first thing to notice is that all identities involving $*$ satisfy
$\varepsilon_{v,w}=\eta_{v}\eta_{w}$. The second thing to notice
is that if $\varepsilon_{T,F}=\eta_{T}\eta_{F}$, then the three remaining
identities follow, because $\varepsilon_{T,F}+\varepsilon_{T,T}+\varepsilon_{T,*}=\eta_{T}$
etc.\begin{eqnarray*}
\varepsilon_{T,F}-\eta_{T}\eta_{F} & = & \delta\left(1-\delta-\alpha\right)-\left(\rho\delta+\left(1-\rho\right)\left(1-\delta-\alpha\right)\right)\left(\left(1-\rho\right)\delta+\rho\left(1-\delta-\alpha\right)\right)\\
 & = & \delta\left(1-\delta-\alpha\right)-\rho\left(1-\rho\right)\left(\delta^{2}+\left(1-\delta-\alpha\right)^{2}\right)-\left(\rho^{2}+\left(1-\rho\right)^{2}\right)\delta\left(1-\delta-\alpha\right)\\
 & = & 2\rho\left(1-\rho\right)\delta\left(1-\delta-\alpha\right)-\rho\left(1-\rho\right)\left(\delta^{2}+\left(1-\delta-\alpha\right)^{2}\right)\\
 & = & -\rho\left(1-\rho\right)\left(2\delta+\alpha-1\right)^{2}\enskip.\end{eqnarray*}

Consequently, $\varepsilon_{T,F}\leq\eta_{T}\eta_{F}$, with equality
iff $\rho\in\left\{ 0,1\right\} $ or $\delta=\frac{1-\alpha}{2}$.

We discard the particular case of $\rho\in\left\{ 0,1\right\} $ (which
corresponds to monotone \ksat, always trivially satisfiable). It
turns out that as soon as $\delta\neq\frac{1-\alpha}{2}$, we could
not make the Second Moment Method work.

On the other hand, when $\delta=\frac{1-\alpha}{2}$ we could make
the Second Moment Method work.
\begin{rem}
Setting $\alpha=0$ corresponds in fact to solutions (which are some
trivial implicants), thus to some extent we get back to the $\delta=\frac{1}{2}$
condition. But is there a positive $\alpha$ yielding a better lower
bound than $\alpha=0$?
\end{rem}

\paragraph{Condition for the Second Moment Method to Work at $\delta=\frac{1-\alpha}{2}$.}

As mentioned in section \ref{sec:Independence-Point}, stationarity
of the independence point implies that $\frac{\partial\Lambda}{\partial\mu}=\frac{\partial\Lambda}{\partial\nu}=\frac{\partial\Lambda}{\partial\pi}=\frac{\partial\Lambda}{\partial\pi^{'}}=0$.

Using equation \ref{eq:der_mu}, i.e. \begin{eqnarray*}
\frac{\partial\Lambda}{\partial\mu} & = & -\sum_{\left(a,b\right)\in D^{2}}\frac{\partial\mu_{a,b}}{\partial\mu}\ln\mu_{a,b}+c\sum_{\left(a,b\right)\in D^{2}}\frac{\partial\mu_{a,b}}{\partial\mu}\sum_{s\in S}\rho_{s}\frac{\Xi_{a\otimes s,b\otimes s}}{\varepsilon_{a\otimes s,b\otimes s}}\enskip,\end{eqnarray*}

we get:\begin{eqnarray*}
\frac{\partial\Lambda}{\partial\mu} & = & \ln\frac{\mu_{*,0}\mu_{0,*}}{\mu_{*,*}\mu_{0,0}}\\
 &  & +c\left(\left(1-\rho\right)\left(\frac{\Xi_{T,T}}{\varepsilon_{T,T}}+\frac{\Xi_{*,*}}{\eta_{*,*}}-\frac{\Xi_{T,*}}{\varepsilon_{T,*}}-\frac{\Xi_{*,T}}{\varepsilon_{*,T}}\right)+\rho\left(\frac{\Xi_{F,F}}{\varepsilon_{F,F}}+\frac{\Xi_{*,*}}{\eta_{*,*}}-\frac{\Xi_{F,*}}{\varepsilon_{F,*}}-\frac{\Xi_{*,F}}{\varepsilon_{*,F}}\right)\right)\enskip;\\
\frac{\partial\Lambda}{\partial\nu} & = & \ln\frac{\mu_{1,0}\mu_{0,1}}{\mu_{1,1}\mu_{0,0}}+c\left(\frac{\Xi_{T,T}}{\varepsilon_{T,T}}+\frac{\Xi_{F,F}}{\varepsilon_{F,F}}-\frac{\Xi_{T,F}}{\varepsilon_{T,F}}-\frac{\Xi_{F,T}}{\varepsilon_{F,T}}\right)\enskip;\\
\frac{\partial\Lambda}{\partial\pi} & = & \ln\frac{\mu_{1,0}\mu_{0,*}}{\mu_{1,*}\mu_{0,0}}\\
 &  & +c\left(\left(1-\rho\right)\left(\frac{\Xi_{T,T}}{\varepsilon_{T,T}}+\frac{\Xi_{F,*}}{\varepsilon_{F,*}}-\frac{\Xi_{F,T}}{\varepsilon_{F,T}}-\frac{\Xi_{T,*}}{\varepsilon_{T,*}}\right)+\rho\left(\frac{\Xi_{F,F}}{\varepsilon_{F,F}}+\frac{\Xi_{T,*}}{\varepsilon_{T,*}}-\frac{\Xi_{T,F}}{\varepsilon_{T,F}}-\frac{\Xi_{F,*}}{\varepsilon_{F,*}}\right)\right)\enskip;\\
\frac{\partial\Lambda}{\partial\pi^{'}} & = & \ln\frac{\mu_{0,1}\mu_{*,0}}{\mu_{*,1}\mu_{0,0}}\\
 &  & +c\left(\left(1-\rho\right)\left(\frac{\Xi_{T,T}}{\varepsilon_{T,T}}+\frac{\Xi_{*,F}}{\varepsilon_{*,F}}-\frac{\Xi_{T,F}}{\varepsilon_{T,F}}-\frac{\Xi_{*,T}}{\varepsilon_{*,T}}\right)+\rho\left(\frac{\Xi_{F,F}}{\varepsilon_{F,F}}+\frac{\Xi_{*,T}}{\varepsilon_{*,T}}-\frac{\Xi_{F,T}}{\varepsilon_{F,T}}-\frac{\Xi_{*,F}}{\varepsilon_{*,F}}\right)\right)\enskip.\end{eqnarray*}

We consider the independence point, thus $\mu_{a,b}=\delta_{a}\delta_{b}$.
Assuming symmetry of occurrences, we may use independence of surfaces
(cf. fact \ref{fac:Indep-of-surfaces}); moreover, using condition
$\varepsilon_{v,w}=\eta_{v}\eta_{w}$, we get that:\begin{eqnarray*}
\frac{\partial\Lambda}{\partial\mu} & = & c\left(\left(1-\rho\right)\left(\frac{\Sigma_{T}}{\eta_{T}}-\frac{\Sigma_{*}}{\eta_{*}}\right)^{2}+\rho\left(\frac{\Sigma_{F}}{\eta_{F}}-\frac{\Sigma_{*}}{\eta_{*}}\right)^{2}\right)\enskip;\\
\frac{\partial\Lambda}{\partial\nu} & = & c\left(\frac{\Sigma_{T}}{\eta_{T}}-\frac{\Sigma_{F}}{\eta_{F}}\right)^{2}\enskip;\\
\frac{\partial\Lambda}{\partial\pi}=\frac{\partial\Lambda}{\partial\pi^{'}} & = & c\left(\left(1-\rho\right)\left(\frac{\Sigma_{T}}{\eta_{T}}-\frac{\Sigma_{F}}{\eta_{F}}\right)\left(\frac{\Sigma_{T}}{\eta_{T}}-\frac{\Sigma_{*}}{\eta_{*}}\right)+\rho\left(\frac{\Sigma_{F}}{\eta_{F}}-\frac{\Sigma_{T}}{\eta_{T}}\right)\left(\frac{\Sigma_{F}}{\eta_{F}}-\frac{\Sigma_{*}}{\eta_{*}}\right)\right)\enskip.\end{eqnarray*}

Canceling out these derivatives yields:

\begin{eqnarray*}
\frac{\Sigma_{*}}{\eta_{*}} & = & \frac{\Sigma_{T}}{\eta_{T}}=\frac{\Sigma_{F}}{\eta_{F}}\enskip.\end{eqnarray*}

Since we assume $\delta=\frac{1-\alpha}{2}$, we have $\eta_{T}=\eta_{F}=\frac{1-\alpha}{2}$;
moreover $\eta_{*}=\alpha$. Thus $\Sigma_{T}=\Sigma_{F}=k\frac{1-\alpha}{2}$
and $\Sigma_{*}=k\alpha$.

It turns out that this condition is sufficient to make the Second
Moment Method work, and numerically we found the critical ratios laid
in table \ref{tab:Critical-ratios-implicants} for standard \threesat{}
with symmetry of occurrences at $\rho=\frac{1}{2}$.

\begin{table}
\caption[Critical ratios of implicants obtained by the Second Moment Method.]{\label{tab:Critical-ratios-implicants}Critical ratios $c$ of implicants
obtained for a given $\alpha$ (and the corresponding choice of the
free $\beta$ parameters). }

\begin{centering}
\begin{tabular}{|c|c|c|c|c|c|}
\hline 
$\alpha$ & $c$ & $\beta_{TFF}$ & $\beta_{TTF}$ & $\beta_{T**}$ & ratio $c$ of \cite{Boufkhad1999}\tabularnewline
\hline
\hline 
$0.001$ & $2.81$ & $0.195$ & $0.10867$ & $3.33\times10^{-5}$ & $4.5$\tabularnewline
\hline 
$0.01$ & $2.77$ & $0.1942$ & $0.098833$ & $3.33\times10^{-5}$ & $4.5$\tabularnewline
\hline 
$0.05$ & $2.52$ & $0.17767$ & $0.08167$ & $0.001633$ & $2$\tabularnewline
\hline 
$0.08$ & $2.32$ & $0.1633$ & $0.07$ & $0.00233$ & $1.5$\tabularnewline
\hline 
$0.11$ & $2.13$ & $0.1533$ & $0.05467$ & $0.0033$ & $1$\tabularnewline
\hline 
$0.15$ & $1.88$ & $0.13833$ & $0.041$ & $0.012$ & -\tabularnewline
\hline 
$0.2$ & $1.59$ & $0.1233$ & $0.02567$ & $0.02833$ & -\tabularnewline
\hline 
$0.25$ & $1.31$ & $0.10833$ & $0.0133$ & $0.04767$ & -\tabularnewline
\hline 
$0.333$ & $0.89$ & $0.094433$ & $3.33\times10^{-5}$ & $0.094167$ & -\tabularnewline
\hline
\end{tabular}
\par\end{centering}

\bigskip{}

\centering{}\includegraphics[width=0.64\columnwidth]{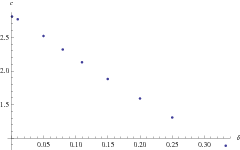}
\end{table}

These values are to be compared with those of Boufkhad \& Dubois -
1999 \cite{Boufkhad1999}, who proved for example that at the ratio
$c=4.5$, any satisfiable instance will have prime implicants with
$\alpha=0.01$. Combined with the lower bound of $3.52$ of \cite{Kaporis2006}
and \cite{Hajiaghayi2003}, this proves that such implicants exist
almost surely when $c\leq3.52$.

Thus in the range $c\in\left(2.81,3.52\right)$ Boufkhad \& Dubois
prove that implicants with $\alpha=0.01$ exist. However, in the range
$c\leq2.81$ the Second Moment Method enable us to establish the existence
of implicants with a $\alpha$ significantly greater than Boufkhad
\& Dubois's.

How can we interpret the fact that the critical $c$ obtained decreases
with $\alpha$? Looking at the set of allowed types of clauses: \begin{eqnarray*}
\mathcal{T} & = & \left\{ TTT,TTF,TFT,FTT,TFF,FTF,FFT,TT*,T*T,*TT,T**,*T*,**T\right\} \\
 &  & \cup\left\{ TF*,T*F,FT*,F*T,*TF,*FT\right\} \enskip,\end{eqnarray*}
we can see that there are $27$ $T$'s, $15$ $F$'s and $15$ $*$'s.
Thus when $*$'s are present, the ratio $T/F$ is $\frac{27}{15}=1.8$.
Without $*$'s, $T/F$ would be $\frac{12}{9}\simeq1.33$ (see section
\ref{sub:Real-True-Surface}). Now, since the Second Moment Method
requires $\Sigma_{T}=\Sigma_{F}$, we can see that it is all the more
artificial as $T/F$ is large. Thus adding $*$'s should cut down
the Second Moment Method's performance.

\subsection{Confrontation of the Second Moment Method with Reality\label{sec:Confrontation-SM-vs-Reality}}

Using \satlab, we investigate the behavior of real solutions and
we emphasize how it differs from the conditions required by the Second
Moment Method that we laid down just above.

Based on numerical calculations of figures \ref{fig:delta1/2-eq}
and \ref{fig:delta1/2-prop}, we conjectured in section \ref{sec:Independence-Point}
that the Second Moment Method might work only if $\varepsilon_{v,w}=\eta_{v}\eta_{w}$.
In this setting we showed that the independence point defined by $\mu_{a,b}=\delta_{a}\delta_{b}$
and $\gamma_{t,u}=\beta_{t}\beta_{u}$ must be a maximum of $T_{2}$,
the second moment. This led us in section \ref{sub:SM-Boolean-Solutions}
to the following necessary condition for the Second Moment Method
to work on boolean solutions: $\delta=\frac{1}{2}$.

Now using \satlab, we are going to give experimental evidence that:
\begin{itemize}
\item real solutions of standard \threesat{} violate all of these conditions:
they are not independent at all!
\item real solutions of standard \threenae{} seem to be rather independent.
\end{itemize}
These observations may explain why the Second Moment Method performs
so poorly on standard \threesat{} (cf. section \ref{sub:SM-Boolean-Solutions})
whereas it works pretty well on \threenae{} (cf. Achlioptas \& Moore
- 2002 \cite{Achlioptas2002}).

\subsubsection{Distances between Solutions \label{sub:Non-Independence-Distances}}

It turns out that in random \threesat, solutions are correlated with
respect to their Hamming distances. Namely their Hamming distances
are not centered around $50\%$ contrary to solutions of random \threenae,
but narrower to each other (cf. figure \ref{fig:SATLab-Hamming-distances}).

What we mean by \emph{Hamming similarity} \index{Hamming similarity}
between two assignments is just the proportion of variables assigned
the same value in both assignments. We took all couples of different
solutions in a sample of random solutions output by a solver, and
we plotted the frequency of the Hamming similarity.

\begin{figure}
\begin{centering}
\includegraphics[width=0.64\textwidth]{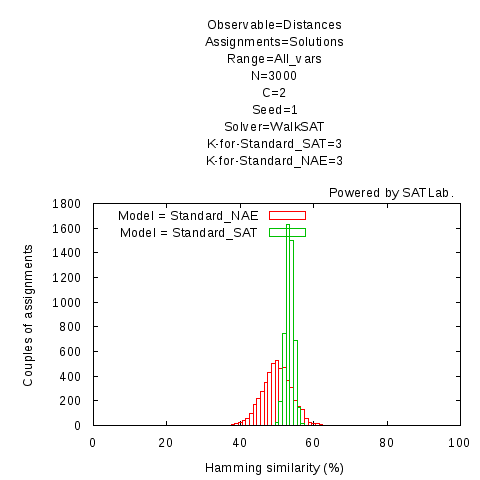}
\par\end{centering}

\caption{\label{fig:SATLab-Hamming-distances}Hamming similarity among solutions
in \nae{} and in \sat.}

\end{figure}

To have a more precise insight into Hamming similarity, we separated
fixed and free variables. Let us recall that a variable is free iff
flipping it yields another solution. We can notice that Hamming similarity
is significantly greater among fixed variables than among free variables
(see figure \ref{fig:SATLab-Hamming-similarity-fixed-vs-free}), and
that it increases with $c$ for both types of variables (see figures
\ref{fig:SATLab-Hamming-similarity-free-c} and \ref{fig:SATLab-Hamming-similarity-fixed-c}).

\begin{figure}
\begin{centering}
\includegraphics[width=0.64\columnwidth]{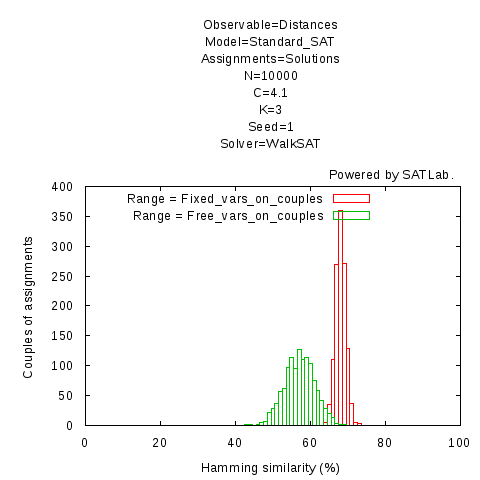}
\par\end{centering}

\caption{\label{fig:SATLab-Hamming-similarity-fixed-vs-free}Hamming similarity
is greater among fixed variables.}

\end{figure}

\begin{figure}
\begin{centering}
\includegraphics[width=0.64\columnwidth]{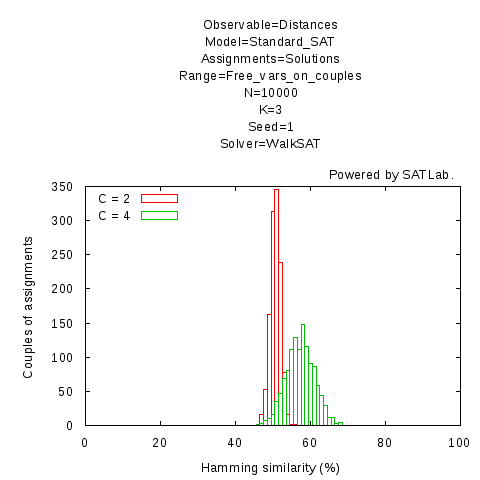}
\par\end{centering}

\caption{\label{fig:SATLab-Hamming-similarity-free-c}Hamming similarity among
free variables increases with $c$.}

\end{figure}

\begin{figure}
\begin{centering}
\includegraphics[width=0.64\columnwidth]{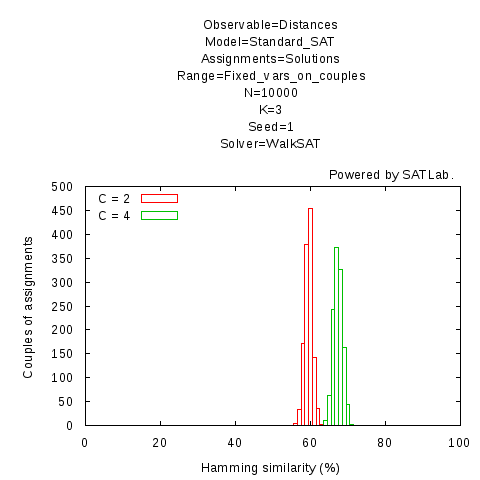}
\par\end{centering}

\caption{\label{fig:SATLab-Hamming-similarity-fixed-c}Hamming similarity among
fixed variables increases with $c$.}

\end{figure}

\clearpage{}

\subsubsection{Surface of True Literals \label{sub:Real-True-Surface}}

What we call true \emph{surface} \index{surface} is the scaled number
of true occurrences of literals. We can see a fundamental difference
between the true surface of fixed variables and the true surface of
free variables. Namely the true surface of fixed variables decreases
with $c$ (figure \ref{fig:SATLab-True-surface-fixed-variables})
whereas the true surface of free variables increases with $c$ (figure
\ref{fig:SATLab-True-surface-free-variables}). Note that both quantities
converge to roughly $0.56$ (i.e. roughly $\frac{4}{7}$) when $c$
approaches the threshold ratio, whereas in section \ref{sub:SM-Boolean-Solutions}
we got the following condition: $\Sigma_{T}=\Sigma_{F}=\frac{k}{2}$
to make the Second Moment Method work.

We interpret the ratio $\frac{4}{7}$ as follows: the allowed types
of clauses are \[
\left\{ TTT,TTF,TFT,FTT,TFF,FTF,FFT\right\} \enskip,\]
which amounts to $12$ $T$'s and $9$ $F$'s. Now $\frac{12}{12+9}=\frac{12}{21}=\frac{4}{7}$.

\begin{figure}
\begin{centering}
\includegraphics[width=0.64\columnwidth]{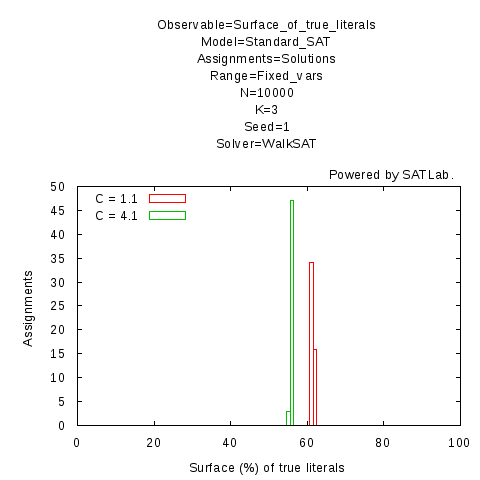}
\par\end{centering}

\caption{\label{fig:SATLab-True-surface-fixed-variables}The true surface of
fixed variables decreases with $c$.}

\end{figure}

\begin{figure}
\begin{centering}
\includegraphics[width=0.64\columnwidth]{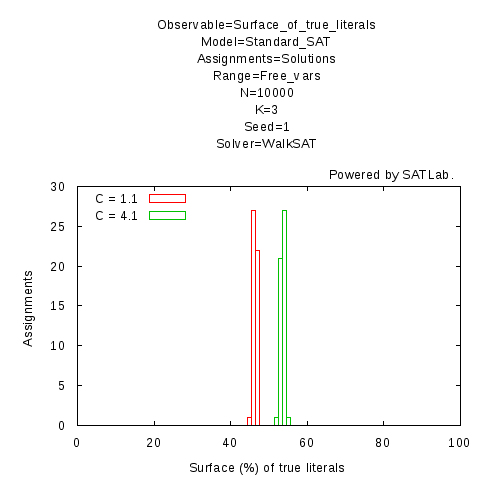}
\par\end{centering}

\caption{\label{fig:SATLab-True-surface-free-variables}The true surface of
free variables increases with $c$.}

\end{figure}

\subsubsection{Non-Independence of True / False Surfaces \label{sub:Non-Independence-of-Surfaces}}

Let us consider two solutions $S_{1}$ and $S_{2}$. We denote by
$\Sigma_{F}$ the false surface under solution $S_{1}$, $\Sigma_{T}$
the true surface under solution $S_{2}$, and $\Xi_{FT}$ the surface
which is false under $S_{1}$ and true under $S_{2}$. In a given
sample of random solutions, we took all couples of different solutions
$\left(S_{1},S_{2}\right)$ and computed the ratio $\frac{\Sigma_{F}\Sigma_{T}}{k\Xi_{FT}}$;
the histogram in figure \ref{fig:SATLab-Independence-of-surfaces}
plots the frequency of this ratio for the solutions of two different
models of formulas: random \threenae{} and random \threesat. Although
some independence seems to exist in \threenae{} (i.e. the ratio is
centered around $1$), it can be seen that there is no independence
of these surfaces for random \threesat.

\begin{figure}
\begin{centering}
\includegraphics[width=0.64\columnwidth]{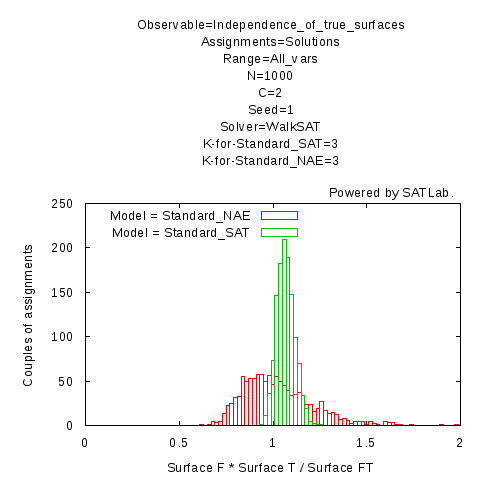}
\par\end{centering}

\caption{\label{fig:SATLab-Independence-of-surfaces}(Non-)independence of
surfaces in \nae{} and in \sat. }

\end{figure}

To have a more precise insight into the non-independence, we separated
fixed and free variables. Let us recall that a variable is free iff
flipping it yields another solution. We can notice that non-independence
comes from both free and fixed variables, but rather from fixed variables
than from free variables, cf. figure \ref{fig:SATLab-Non-independence-of-surfaces-from-fixed}.

\begin{figure}
\begin{centering}
\includegraphics[width=0.64\columnwidth]{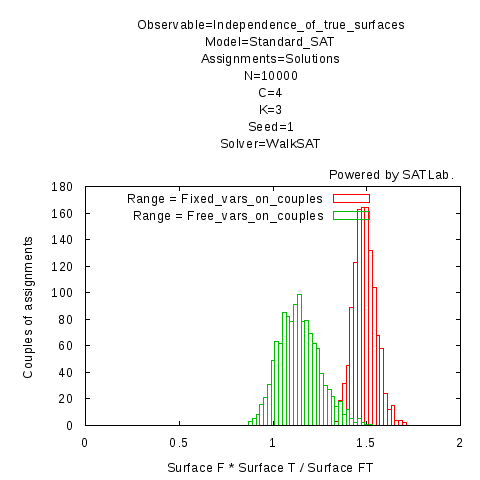}
\par\end{centering}

\caption{\label{fig:SATLab-Non-independence-of-surfaces-from-fixed}The non-independence
of surfaces comes rather from fixed variables.}

\end{figure}

\subsubsection{Non-Independence of Clauses Types}

Here what we call \emph{clause type}\index{clause type} is the number
of true occurrences of variables in the clause. Let us consider two
solutions $S_{1}$ and $S_{2}$. We denote by $\beta_{1}$ the proportion
of uniquely satisfied clauses under solution $S_{1}$, $b_{1}$ the
proportion of uniquely satisfied clauses under solution $S_{2}$,
and $\gamma_{1,1}$ the proportion of clauses which are uniquely satisfied
under $S_{1}$ and under $S_{2}$. In a given sample of random solutions,
we took all couples of different solutions $\left(S_{1},S_{2}\right)$
and computed the ratio $\frac{\beta_{1}b_{1}}{\gamma_{1,1}}$; the
histogram in figure \ref{fig:SATLab-independence-of-unisat-clauses}
plots the frequency of this ratio for the solutions of two different
kinds of assignments: solutions of a random \threenae{} formula and
solutions of a random \threesat{} formula. Although independence
seems to happen among solutions of random \threenae{} (i.e. the ratio
is centered around $1$), it can be seen that there is no independence
in random \threesat.

In real instances we can assume symmetry of occurrences, in the sense
of section \ref{sub:Symmetry-of-Occurrences}; so in the light of
fact \ref{fac:Indep-of-surfaces}, we could conclude from the non-independence
of surfaces observed in section \ref{sub:Non-Independence-of-Surfaces}
that in the real solutions of \threesat{} independence of clauses
types would not hold.

\begin{figure}
\begin{centering}
\includegraphics[width=0.64\columnwidth]{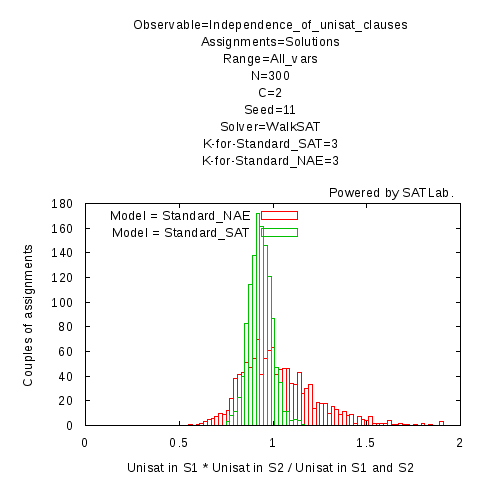}
\par\end{centering}

\caption{\label{fig:SATLab-independence-of-unisat-clauses}(Non-)independence
of uniquely satisfied clauses in \nae{} and in \sat.}

\end{figure}

\section{Second Moment Method on Distributional Random \ksat\label{cha:Distributional-SM}}

Using the standard distributional model instead of the standard drawing
model yields better upper bounds on the satisfiability threshold.
Moreover, we would like to gain some more control over the proportion
of variables assigned $1$ according to the imbalance between their
positive and negative occurrences. Namely, a variable is all the more
expected to be assigned $1$ in a solution as it has more positive
occurrences, and vice-versa. At least this seems to happen on real
solutions, see figure \ref{fig:SATLab-Indep-deltatp-c2} in section
\ref{sub:dpq-Non-Independence-of-Values}.

That is the reasons why we are going to implement the Second Moment
Method in the distributional model.

In this section we follow roughly the same outline than in section
\ref{cha:General-Framework-SMM}, but we focus on solutions only;
we only emphasize the major differences with respect to the general
framework of section \ref{cha:General-Framework-SMM}.

\subsection{Preliminaries}

\subsubsection{Occurrences and Signs}

We still have $n$ \emph{variables}. We denote by $d_{p,q}$ the fraction
of variables having $p$ positive and $q$ negative occurrences.

\begin{eqnarray*}
\sum_{\left(p,q\right)\in\boldsymbol{N}^{2}}d_{p,q} & = & 1\enskip;\\
\sum_{\left(p,q\right)\in\boldsymbol{N}^{2}}\left(p+q\right)d_{p,q} & = & kc\enskip.\end{eqnarray*}

Occurrences and signs of variables are determined \emph{a priori}.

We ought to consider light and heavy variables $\mathcal{L}$ and
$\mathcal{H}$ (cf. \cite{Diaz2009}). In fact we are going not to
worry about that, because they make the calculation heavier, and in
the end we shall see that there is no need to be rigorous since we
only have negative results.

\subsubsection{Values}

Given a boolean assignment, we denote by $\delta_{p,q}$ the proportion
of variables with $p$ positive and $q$ negative occurrences which
are assigned $1$. Thus the proportion of variables with $p$ positive
and $q$ negative occurrences which are assigned $0$ is $1-\delta_{p,q}$.

Given two assignments $S_{1}$ and $S_{2}$, for all $\left(a,b\right)\in D^{2}$,
considering variables with $p$ positive and $q$ negative occurrences,
we denote by:
\begin{itemize}
\item $\lambda_{p,q}$ the proportion of variables which are assigned $1$
in $S_{1}$ and $S_{2}$;
\item $\mu_{p,q}$ the proportion of variables which are assigned $1$ in
$S_{1}$ and $0$ in $S_{2}$;
\item $\mu_{p,q}^{'}$ the proportion of variables which are assigned $0$
in $S_{1}$ and $1$ in $S_{2}$;
\item $\nu_{p,q}$ the proportion of variables which are assigned $0$ in
$S_{1}$ and $S_{2}$.
\end{itemize}
We have the following constraints:

\begin{eqnarray*}
\delta_{p,q} & = & \lambda_{p,q}+\mu_{p,q}\\
\delta_{p,q} & = & \lambda_{p,q}+\mu_{p,q}^{'}\\
1-\delta_{p,q} & = & \mu_{p,q}+\nu_{p,q}\\
1-\delta_{p,q} & = & \mu_{p,q}^{'}+\nu_{p,q}\end{eqnarray*}

thus\begin{eqnarray*}
\lambda_{p,q} & = & \delta_{p,q}-\mu_{p,q}\\
\mu_{p,q}^{'} & = & \mu_{p,q}\\
\nu_{p,q} & = & 1-\delta_{p,q}-\mu_{p,q}\end{eqnarray*}

which enables us to work only with $\mu_{p,q}$.

\subsubsection{Truth Values }

\begin{eqnarray*}
\eta_{T} & = & \frac{1}{kc}\sum_{\left(p,q\right)\in\boldsymbol{N}^{2}}\left(p\delta_{p,q}+q\left(1-\delta_{p,q}\right)\right)d_{p,q}\\
\eta_{F} & = & \frac{1}{kc}\sum_{\left(p,q\right)\in\boldsymbol{N}^{2}}\left(q\delta_{p,q}+p\left(1-\delta_{p,q}\right)\right)d_{p,q}\\
\varepsilon_{T,T} & = & \frac{1}{kc}\sum_{\left(p,q\right)\in\boldsymbol{N}^{2}}\left(p\left(\delta_{p,q}-\mu_{p,q}\right)+q\left(1-\delta_{p,q}-\mu_{p,q}\right)\right)d_{p,q}\\
\varepsilon_{F,F} & = & \frac{1}{kc}\sum_{\left(p,q\right)\in\boldsymbol{N}^{2}}\left(q\left(\delta_{p,q}-\mu_{p,q}\right)+p\left(1-\delta_{p,q}-\mu_{p,q}\right)\right)d_{p,q}\\
\varepsilon_{T,F}=\varepsilon_{F,T} & = & \frac{1}{kc}\sum_{\left(p,q\right)\in\boldsymbol{N}^{2}}\left(p+q\right)\mu_{p,q}d_{p,q}\end{eqnarray*}

\subsubsection{Types of Clauses and Surfaces}

We keep the same definitions as in section \ref{sub:General-Framework-Preliminaries}.

However, some extra constraints on surfaces occur in the distributional
model, because here all occurrences and signs of variables are determined
\emph{a priori}:\begin{eqnarray}
\Sigma_{v} & = & k\eta_{v}\enskip;\label{eq:constr-surface-FM}\\
\Xi_{v,w} & = & k\varepsilon_{v,w}\enskip.\label{eq:constr-surface-SM}\end{eqnarray}

\subsection{Expression of the First Moment}

The first moment of the number $X$ of solutions can be split up into
the following factors: total number of assignments and probability
for an assignment to be a solution.
\begin{enumerate}
\item total number of assignments: choose subsets of variables assigned
$0$ or $1$:\\
$\prod_{\left(p,q\right)\in\mathcal{L}}{d_{p,q}n \choose \delta_{p,q}d_{p,q}n}$;
\item probability for an assignment to be a solution:

\begin{enumerate}
\item number of satisfied formulas:

\begin{enumerate}
\item we give each clause an allowed type $t\in\mathcal{T}$: ${cn \choose \dots\left(\beta_{t}cn\right)_{t\in\mathcal{T}}\dots}$
\item we find a permutation of the true literals into the true boxes and
a permutation of the false literals into the false boxes: $\left(\eta_{T}kcn\right)!\left(\eta_{F}kcn\right)!$
\end{enumerate}
\item total number of formulas, i.e. number of permutations of the occurrences
of literals into the boxes: $\left(kcn\right)!$
\end{enumerate}
\end{enumerate}
We denote by $\mathcal{P}$ the set of all families of non-negative
numbers $\left(\left(\delta_{p,q}\right)_{\left(p,q\right)\in\mathcal{L}},\left(\beta_{t}\right)_{t\in\mathcal{T}}\right)$
satisfying constraint \ref{eq:sumbetas1}. We denote by $\mathcal{I}\left(n\right)$
the intersection of $\mathcal{P}$ with the multiples of $\frac{1}{n}$;
we get the following expression of the first moment:

\begin{eqnarray*}
\mathrm{E}X & = & \sum_{\left(\left(\delta_{p,q}\right)_{\left(p,q\right)\in\mathcal{L}},\left(\beta_{t}\right)_{t\in\mathcal{T}}\right)\in\mathcal{I}\left(n\right)}T_{1}\left(n\right)\end{eqnarray*}

where

\begin{eqnarray*}
T_{1}\left(n\right) & = & \frac{\prod_{\left(p,q\right)\in\mathcal{L}}{d_{p,q}n \choose \delta_{p,q}d_{p,q}n}{cn \choose \dots\left(\beta_{t}cn\right)_{t\in\mathcal{T}}\dots}}{{kcn \choose \eta_{T}kcn}}\enskip.\end{eqnarray*}

The exponential equivalent of $T_{1}\left(n\right)$ is $T_{1}^{n}$,
where

\begin{eqnarray*}
T_{1} & = & \frac{1}{\prod_{\left(p,q\right)\in\mathcal{L}^{2}}\delta_{p,q}^{\delta_{p,q}d_{p,q}}}\left(\frac{\left(\eta_{T}^{\eta_{T}}\eta_{F}^{\eta_{F}}\right)^{k}}{\prod_{t\in\mathcal{T}}\beta_{t}^{\beta_{t}}}\right)^{c}\enskip.\end{eqnarray*}

\subsection{Expression of the Second Moment}

The second moment of $X$ can be split up into the following factors:
total number of assignments and probability for an assignment to be
a solution.
\begin{enumerate}
\item total number of assignments: choose subsets of variables assigned
$0$ or $1$:\\
$\prod_{\left(p,q\right)\in\mathcal{L}}{d_{p,q}n \choose \lambda_{p,q}d_{p,q}n,\mu_{p,q}d_{p,q}n,\mu_{p,q}d_{p,q}n,\nu_{p,q}d_{p,q}n}$;
\item probability for an assignment to be a solution:

\begin{enumerate}
\item number of satisfied formulas:

\begin{enumerate}
\item we give each clause two allowed types: ${cn \choose \dots\left(\gamma_{t,u}cn\right)_{\left(t,u\right)\in\mathcal{T}^{2}}\dots}$
\item we find a permutation of the literals into the corresponding boxes:\\
$\left(\varepsilon_{T,T}kcn\right)!\left(\varepsilon_{T,F}kcn\right)!\left(\varepsilon_{F,T}kcn\right)!\left(\varepsilon_{F,F}kcn\right)!$
\end{enumerate}
\item total number of formulas, i.e. number of permutations of the occurrences
of literals into the boxes: $\left(kcn\right)!$
\end{enumerate}
\end{enumerate}
We denote by $\mathcal{P}_{2}$ the set of all families of non-negative
numbers $\left(\left(\mu_{p,q}\right)_{\left(p,q\right)\in\mathcal{L}},\left(\gamma_{t,u}\right)_{\left(t,u\right)\in\mathcal{T}^{2}}\right)$
satisfying constraints \ref{eq:beta_u} and \ref{eq:beta_t}. We denote
by $\mathcal{I}_{2}\left(n\right)$ the intersection of $\mathcal{P}_{2}$
with the multiples of $\frac{1}{n}$; we get the following expression
of the second moment:

\begin{eqnarray*}
\mathrm{E}X^{2} & = & \sum_{\left(\left(\mu_{p,q}\right)_{\left(p,q\right)\in\mathcal{L}},\left(\gamma_{t,u}\right)_{\left(t,u\right)\in\mathcal{T}^{2}}\right)\in\mathcal{I}_{2}\left(n\right)}T_{2}\left(n\right)\end{eqnarray*}

where

\begin{eqnarray*}
T_{2}\left(n\right) & = & \frac{\prod_{\left(p,q\right)\in\mathcal{L}}{d_{p,q}n \choose \lambda_{p,q}d_{p,q}n,\mu_{p,q}d_{p,q}n,\mu_{p,q}d_{p,q}n,\nu_{p,q}d_{p,q}n}{cn \choose \dots\left(\gamma_{t,u}cn\right)_{\left(t,u\right)\in\mathcal{T}^{2}}\dots}}{{kcn \choose \varepsilon_{T,T}kcn,\varepsilon_{T,F}kcn,\varepsilon_{F,T}kcn,\varepsilon_{F,F}kcn}}\enskip.\end{eqnarray*}

The exponential equivalent of $T_{2}\left(n\right)$ is $T_{2}^{n}$,
where

\begin{eqnarray*}
T_{2} & = & \frac{1}{\prod_{\left(p,q\right)\in\mathcal{L}}\left(\left(\delta_{p,q}-\mu_{p,q}\right)^{\delta_{p,q}-\mu_{p,q}}\mu_{p,q}^{2\mu_{p,q}}\left(1-\delta_{p,q}-\mu_{p,q}\right)^{1-\delta_{p,q}-\mu_{p,q}}\right)^{d_{p,q}}}\\
 &  & \cdot\left(\frac{\left(\varepsilon_{T,T}^{\varepsilon_{T,T}}\varepsilon_{T,F}^{\varepsilon_{T,F}}\varepsilon_{F,T}^{\varepsilon_{F,T}}\varepsilon_{F,F}^{\varepsilon_{F,F}}\right)^{k}}{\prod_{\left(t,u\right)\in\mathcal{T}^{2}}\gamma_{t,u}^{\gamma_{t,u}}}\right)^{c}\enskip.\end{eqnarray*}

\subsection{Expression of the Lagrangian\label{sub:SM-dpq-Lagrangian}}

When the parameters of the first moment (i.e. $\left(\delta_{p,q}\right)_{\left(p,q\right)\in\mathcal{L}},\left(\beta_{t}\right)_{t\in\mathcal{T}}$)
are chosen, $T_{2}$ must be maximized under constraints \ref{eq:beta_u}
and \ref{eq:beta_t}. That leads us to use the Lagrange multipliers
method. In order to make the forthcoming maximization easier, we introduce
some extra variables $\psi_{v,w}$ which are going to simulate $\varepsilon_{v,w}$.
The reason for this is that $\varepsilon_{v,w}$ contains $\mu_{p,q}$,
but we need the expression of $\mu_{p,q}$ for our numerical calculations.
So, because of equation \ref{eq:constr-surface-SM}, we have the following
constraints:

\begin{eqnarray*}
k\varepsilon_{T,T}=k\psi_{T,T} & = & \Xi_{T,T}\\
k\varepsilon_{T,F}=k\psi_{T,F} & = & \Xi_{T,F}\\
k\varepsilon_{F,T}=k\psi_{F,T} & = & \Xi_{F,T}\\
k\varepsilon_{F,F}=k\psi_{F,F} & = & \Xi_{F,F}\end{eqnarray*}

Using the facts that $\varepsilon_{T,T}+\varepsilon_{T,F}=\eta_{T}$,
$\Xi_{T,T}+\Xi_{T,F}=\Sigma_{T}$ and $\Sigma_{T}=k\eta_{T}$, we
see that constraint $\varepsilon_{T,T}=\psi_{T,T}$ is redundant.
Eliminating $\varepsilon_{F,T}$ and $\varepsilon_{F,F}$ as well,
there remain the following 5 constraints:

\begin{eqnarray}
\varepsilon_{T,F} & = & \psi_{T,F}\label{eq:c1}\\
\psi_{T,T} & = & \eta_{T}-\psi_{T,F}\label{eq:c2}\\
\psi_{T,F} & = & \frac{\Xi_{T,F}}{k}\label{eq:c3}\\
\psi_{F,T} & = & \psi_{T,F}\label{eq:c4}\\
\psi_{F,F} & = & \eta_{F}-\psi_{T,F}\label{eq:c5}\end{eqnarray}

So we define the following Lagrangian:

\begin{eqnarray*}
\Lambda & = & -\sum_{\left(p,q\right)\in\mathcal{L}}d_{p,q}\left(\left(\delta_{p,q}-\mu_{p,q}\right)\ln\frac{\delta_{p,q}-\mu_{p,q}}{e}+\left(1-\delta_{p,q}-\mu_{p,q}\right)\ln\frac{1-\delta_{p,q}-\mu_{p,q}}{e}\right)\\
 &  & -2\sum_{\left(p,q\right)\in\mathcal{L}}d_{p,q}\mu_{p,q}\ln\frac{\mu_{p,q}}{e}-c\sum_{\left(t,u\right)\in\mathcal{T}^{2}}\gamma_{t,u}\ln\frac{\gamma_{t,u}}{e}+kc\sum_{\left(v,w\right)\in\left\{ T,F\right\} ^{2}}\psi_{v,w}\ln\frac{\psi_{v,w}}{e}\\
 &  & +c\sum_{\left(v,w\right)\in\left\{ T,F\right\} ^{2}}\left(\ln h_{v,w}\right)\left(\Xi_{v,w}-k\psi_{v,w}\right)+kc\left(\ln y\right)\left(\psi_{T,F}-\varepsilon_{T,F}\right)\\
 &  & +c\sum_{t\in\mathcal{T}}\left(\ln f_{t}\right)\left(\sum_{u\in\mathcal{T}}\gamma_{t,u}-\beta_{t}\right)+c\sum_{u\in\mathcal{U}}\left(\ln g_{u}\right)\left(\sum_{t\in\mathcal{T}}\gamma_{t,u}-\beta_{u}\right)\enskip.\end{eqnarray*}

\subsubsection{Derivative with respect to $\gamma_{t,u}$}

\begin{eqnarray*}
\frac{\partial\Lambda}{\partial\gamma_{t,u}} & = & -c\ln\gamma_{t,u}+c\sum_{i=1}^{k}\ln h_{t_{i},u_{i}}+c\ln f_{t}+c\ln g_{u}\enskip.\end{eqnarray*}

Canceling out this derivative yields:

\begin{eqnarray}
\gamma_{t,u} & = & f_{t}g_{u}\prod_{i=1}^{k}h_{t_{i},u_{i}}\enskip.\label{eq:Distributional-gamma_tu}\end{eqnarray}

\subsubsection{Derivative with respect to $\psi_{v,w}$}

\begin{eqnarray*}
\frac{\partial\Lambda}{\partial\psi_{v,w}} & = & kc\ln\psi_{v,w}-kc\ln h_{v,w}+kc\left(\ln y\right)\boldsymbol{1}_{v=T\land w=F}\enskip.\end{eqnarray*}

Canceling out these derivatives yields:

\begin{eqnarray}
\psi_{v,w} & = & h_{v,w}y^{-\boldsymbol{1}_{v=T\land w=F}}\enskip.\label{eq:psi_vw}\end{eqnarray}

Thus constraints \ref{eq:c1}, \ref{eq:c2}, \ref{eq:c3}, \ref{eq:c4}
and \ref{eq:c5} become:\begin{eqnarray*}
\varepsilon_{T,F} & = & \frac{h_{T,F}}{y}\\
h_{T,T} & = & \eta_{T}-\frac{h_{T,F}}{y}\\
\frac{h_{T,F}}{y} & = & \frac{\Xi_{T,F}}{k}\\
h_{F,T} & = & \frac{h_{T,F}}{y}\\
h_{F,F} & = & \eta_{F}-\frac{h_{T,F}}{y}\end{eqnarray*}

\subsubsection{Derivative with respect to $\mu_{p,q}$\label{sub:Derivative-mu-pq}}

\begin{eqnarray}
\frac{\partial\Lambda}{\partial\mu_{p,q}} & = & d_{p,q}\ln\frac{\left(\delta_{p,q}-\mu_{p,q}\right)\left(1-\delta_{p,q}-\mu_{p,q}\right)}{\mu_{p,q}^{2}}-\left(p+q\right)d_{p,q}\ln y\enskip.\label{eq:der-mupq}\end{eqnarray}

Canceling out this derivative yields:

\begin{eqnarray*}
\left(\delta_{p,q}-\mu_{p,q}\right)\left(1-\delta_{p,q}-\mu_{p,q}\right) & = & \mu_{p,q}^{2}y^{p+q}\enskip,\end{eqnarray*}

i.e.\begin{eqnarray*}
\mu_{p,q}^{2}\left(1-y^{p+q}\right)-\mu_{p,q}+\delta_{p,q}\left(1-\delta_{p,q}\right) & = & 0\enskip.\end{eqnarray*}

Thus there are 2 cases to consider:
\begin{enumerate}
\item case where $y=1$ or $p+q=0$: $\mu_{p,q}=\delta_{p,q}\left(1-\delta_{p,q}\right)$;
\item case where $y\neq1$ and $p+q\neq0$: $\mu_{p,q}=\frac{1\pm\sqrt{1-4\left(1-y^{p+q}\right)\delta_{p,q}\left(1-\delta_{p,q}\right)}}{2\left(1-y^{p+q}\right)}$;
numerically we can find solutions with $\mu_{p,q}=\frac{1-\sqrt{1-4\left(1-y^{p+q}\right)\delta_{p,q}\left(1-\delta_{p,q}\right)}}{2\left(1-y{}^{p+q}\right)}$.
\end{enumerate}

\subsection{Independence Point}

As in section \ref{sec:Independence-Point}, we define the independence
point by $\mu_{p,q}=\delta_{p,q}\left(1-\delta_{p,q}\right)$ and
$\gamma_{t,u}=\beta_{t}\beta_{u}$. Again, we were able to make the
Second Moment Method work only if $\varepsilon_{v,w}=\eta_{v}\eta_{w}$.

When $\varepsilon_{v,w}=\eta_{v}\eta_{w}$, we have $\frac{T_{2}}{T_{1}^{2}}=1$
(see proof of theorem \ref{thm:1-at-independence}) and the independence
point is stationary without any extra condition (plug $h_{v,w}=\eta_{v}\eta_{w}$,
$y=1$, $f_{t}=\frac{\beta_{t}}{\prod_{i=1}^{k}\eta_{t_{i}}}$, $g_{u}=\frac{\beta_{u}}{\prod_{i=1}^{k}\eta_{u_{i}}}$
into the constraints and equations \ref{eq:Distributional-gamma_tu},
\ref{eq:psi_vw} and \ref{eq:der-mupq}), assuming symmetry of occurrences
as usual (and thus fact \ref{fac:Indep-of-surfaces}). By comparison
with chapter \ref{cha:General-Framework-SMM}, we could say that the
extra condition we had there on the surfaces to make the independence
point stationary corresponds here to the preliminary extra constraint
\ref{eq:constr-surface-SM}.

Moreover, it is noteworthy that, because of fact \ref{fac:Indep-of-surfaces}
and constraint \ref{eq:constr-surface-FM}, the independence point
violates constraint \ref{eq:constr-surface-SM} when $\varepsilon_{v,w}\neq\eta_{v}\eta_{w}$.

So, what is the condition for $\varepsilon_{v,w}=\eta_{v}\eta_{w}$?

\subsubsection{Condition for $\varepsilon_{v,w}=\eta_{v}\eta_{w}$ at the Independence
Point\label{sub:Condition-for-vw-dpq}}

As before in section \ref{sub:SM-Boolean-Solutions}, the first thing
to notice is that if $\varepsilon_{T,F}=\eta_{T}\eta_{F}$, then the
three other identities follow, because $\varepsilon_{T,F}+\varepsilon_{T,T}=\eta_{T}$
etc.

We are grateful to Emmanuel Lepage, who gave us the main idea to compare
$\varepsilon_{T,F}$ and $\eta_{T}\eta_{F}$. Let us make the following
change of variables: $\delta_{p,q}=\frac{1}{2}+\delta_{p,q}^{'}$:

\begin{eqnarray*}
\eta_{T} & = & \frac{1}{2}+\frac{1}{kc}\sum_{\left(p,q\right)\in\boldsymbol{N}^{2}}\left(p-q\right)\delta_{p,q}^{'}d_{p,q}\enskip;\\
\eta_{F} & = & \frac{1}{2}-\frac{1}{kc}\sum_{\left(p,q\right)\in\boldsymbol{N}^{2}}\left(p-q\right)\delta_{p,q}^{'}d_{p,q}\enskip;\\
\varepsilon_{T,F} & = & \frac{1}{4}-\frac{1}{kc}\sum_{\left(p,q\right)\in\boldsymbol{N}^{2}}\left(p+q\right)\delta_{p,q}^{'2}d_{p,q}\enskip.\end{eqnarray*}

Thus

\begin{eqnarray*}
\varepsilon_{T,F}-\eta_{T}\eta_{F} & = & \left(\frac{1}{kc}\sum_{\left(p,q\right)\in\boldsymbol{N}^{2}}\left(p-q\right)\delta_{p,q}^{'}d_{p,q}\right)^{2}-\frac{1}{kc}\sum_{\left(p,q\right)\in\boldsymbol{N}^{2}}\left(p+q\right)\delta_{p,q}^{'2}d_{p,q}\enskip.\end{eqnarray*}

\begin{enumerate}
\item $\left(\frac{1}{kc}\sum_{\left(p,q\right)\in\boldsymbol{N}^{2}}\left(p-q\right)\delta_{p,q}^{'}d_{p,q}\right)^{2}\leq\left(\frac{1}{kc}\sum_{\left(p,q\right)\in\boldsymbol{N}^{2}}\left|\left(p-q\right)\delta_{p,q}^{'}\right|d_{p,q}\right)^{2}$,
with equality iff\\
$\left(p-q\right)\delta_{p,q}^{'}$ has the same sign wherever
$d_{p,q}\neq0$;
\item Since $p,q\geq0$, $\left|p-q\right|\leq p+q$ with equality iff $p=0$
or $q=0$;\\
thus $\left(\frac{1}{kc}\sum_{\left(p,q\right)\in\boldsymbol{N}^{2}}\left|\left(p-q\right)\delta_{p,q}^{'}\right|d_{p,q}\right)^{2}\leq\left(\frac{1}{kc}\sum_{\left(p,q\right)\in\boldsymbol{N}^{2}}\sqrt{p+q}^{2}\left|\delta_{p,q}^{'}\right|\sqrt{d_{p,q}}^{2}\right)^{2}$,
with equality iff $\delta_{p,q}^{'}=0$ wherever $p\neq0$, $q\neq0$
and $d_{p,q}\neq0$;
\item By the Cauchy-Schwartz inequality,\begin{eqnarray*}
 &  & \left(\frac{1}{kc}\sum_{\left(p,q\right)\in\boldsymbol{N}^{2}}\sqrt{p+q}^{2}\left|\delta_{p,q}^{'}\right|\sqrt{d_{p,q}}^{2}\right)^{2}\\
 & \leq & \left(\frac{1}{kc}\sum_{\left(p,q\right)\in\boldsymbol{N}^{2}}\sqrt{p+q}^{2}\sqrt{d_{p,q}}^{2}\right)\left(\frac{1}{kc}\sum_{\left(p,q\right)\in\boldsymbol{N}^{2}}\sqrt{p+q}^{2}\left|\delta_{p,q}^{'}\right|^{2}\sqrt{d_{p,q}}^{2}\right)\\
 & = & \frac{1}{kc}\sum_{\left(p,q\right)\in\boldsymbol{N}^{2}}\left(p+q\right)\delta_{p,q}^{'2}d_{p,q}\enskip,\end{eqnarray*}
with equality iff $\left|\delta_{p,q}^{'}\right|$ has the same value
wherever $\left(p+q\right)d_{p,q}\neq0$.
\end{enumerate}
To conclude, $\varepsilon_{T,F}\leq\eta_{T}\eta_{F}$ with equality
iff ($\delta_{p,q}^{'}=0$ whenever $\left(p+q\right)d_{p,q}\neq0$)
or ($\delta_{p,q}^{'}$ is symmetric in $p,q$ and the model has only
pure literals).

This means that in all models allowing non-pure literals (in particular
the standard model having a 2D-Poisson $d_{p,q}$), $\varepsilon_{T,F}=\eta_{T}\eta_{F}$
iff $\delta_{p,q}=\frac{1}{2}$ whenever $\left(p+q\right)d_{p,q}\neq0$.

Consequently, even in the distributional model, we encounter the very
restrictive condition $\delta_{p,q}=\frac{1}{2}$ to make the Second
Moment Method work.

Numerically, we found a critical ratio of $2.838$, thus very slightly
above the $2.833$ obtained in the drawing model (cf section \ref{sub:SM-Boolean-Solutions}).

\subsubsection{Attempts to evade the $\delta_{p,q}=\frac{1}{2}$ condition \label{sub:Attempts-to-evade-dpq-1/2}}

The shape of $\delta_{p,q}$ on figures \ref{fig:SATLab-Indep-deltatp-c2}
and \ref{fig:SATLab-Non-indep-deltatp-c4} suggested us that on real
solutions $\delta_{p,q}=$$\frac{1}{1+\omega^{p-q}}$. So we tried
to evade the $\omega=1$ case.

We plotted $\ln\ln\frac{F_{2}}{F_{1}^{2}}$ for different values of
$\omega$, at a point satisfying constraints \ref{eq:beta_u}, \ref{eq:beta_t},
\ref{eq:constr-surface-SM}, \ref{eq:c1}, \ref{eq:c2}, \ref{eq:c3},
\ref{eq:c4} and \ref{eq:c5}, for the best choice of the $\beta_{t}$'s
that we found complying with constraint \ref{eq:constr-surface-FM}.
The expected value is $-\infty$ iff $\frac{F_{2}}{F_{1}^{2}}=1$.
We set the ratio $c=0.1$ (to be compared with $2.838$, where $\omega=1$
works). Only $\omega=1$ seems to make $\frac{F_{2}}{F_{1}^{2}}=1$,
cf. figure \ref{fig:delta_tp-omega=00003D1}.

\begin{figure}
\begin{centering}
\includegraphics[width=0.64\columnwidth]{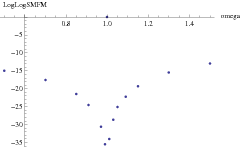}
\par\end{centering}

\caption{\label{fig:delta_tp-omega=00003D1}Setting $\delta_{p,q}=$$\frac{1}{1+\omega^{p-q}}$
makes the Second Moment Method work iff $\omega=1$.}

\end{figure}

\subsection{Confrontation with Reality \label{sec:dpq-Confrontation-with-Reality}}

We are going to do the same kinds of observations through \satlab{}
as in section \ref{sec:Confrontation-SM-vs-Reality} in order to figure
out why the Second Moment Method still fails to give high upper bounds
in the distributional model.

\subsubsection{Non-Independence of Values \label{sub:dpq-Non-Independence-of-Values}}

We focus our attention on variables with $T$ occurrences among which
$U$ are positive. We denote by $d$ the average proportion of those
variables assigned $1$ by a solution and $u$ the average proportion
of those variables assigned $0$ and $1$ by a couple of distinct
solutions. In a given sample of random solutions, we took all couples
of different solutions and computed the following three quantities:
$d$, $d\_d=d\left(1-d\right)$ and $u$. At independence we should
have $u=d\_d$, which happens for $c=2$ (cf. figure \ref{fig:SATLab-Indep-deltatp-c2})
but not for $c=4$ (cf. figure \ref{fig:SATLab-Non-indep-deltatp-c4}).

Moreover we can see that $d$ is almost linear in $U$ when $c=2$
but it curves when $c=4$. Note also that the range of $U$ may be
strictly included in $\left[0\dots T\right]$ (cf. figure \ref{fig:SATLab-Non-indep-deltatp-c4}).
Determining the shape of $d$ might help do better calculations of
first and second moments, even though in section \ref{sub:Attempts-to-evade-dpq-1/2}
we took $\delta_{p,q}=$$\frac{1}{1+\omega^{p-q}}$ but it could not
make the Second Moment Method work. It is clear however that the condition
$\delta_{p,q}=\frac{1}{2}$ we encountered in section \ref{sub:Condition-for-vw-dpq}
does not hold on real solutions.

\begin{figure}
\begin{centering}
\includegraphics[width=0.64\columnwidth]{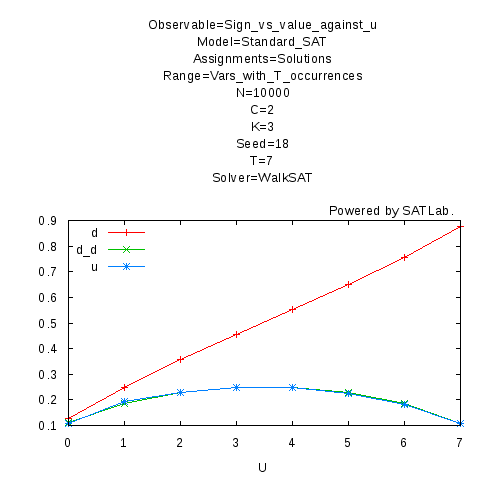}
\par\end{centering}

\caption{\label{fig:SATLab-Indep-deltatp-c2}$\delta_{p,q}$ and $\mu_{p,q}=\delta_{p,q}\left(1-\delta_{p,q}\right)$
at $c=2$.}

\end{figure}

\begin{figure}
\begin{centering}
\includegraphics[width=0.64\columnwidth]{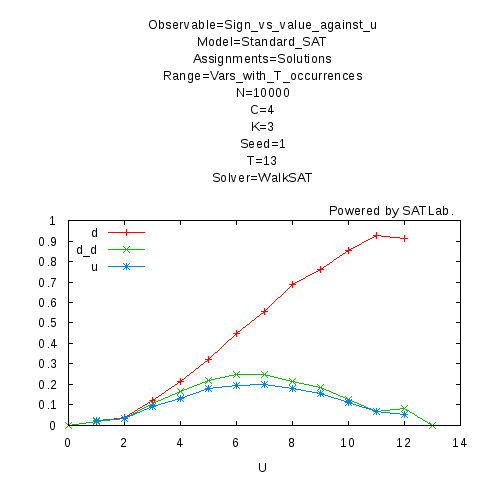}
\par\end{centering}

\caption{\label{fig:SATLab-Non-indep-deltatp-c4}$\delta_{p,q}$ and $\mu_{p,q}\neq\delta_{p,q}\left(1-\delta_{p,q}\right)$
at $c=4$.}

\end{figure}

\subsubsection{Non-Independence of Surfaces}

We perform the same experiment as in section \ref{sub:Non-Independence-of-Surfaces},
but we restrict surfaces to variables having $T$ occurrences among
which $U$ are positive. On figure \ref{fig:SATLab-Non-independence-T-U}
we can see that there is still no independence of surfaces, although
the restriction of the surfaces to these variables curbs the non-independence.

\begin{figure}
\begin{centering}
\includegraphics[width=0.64\columnwidth]{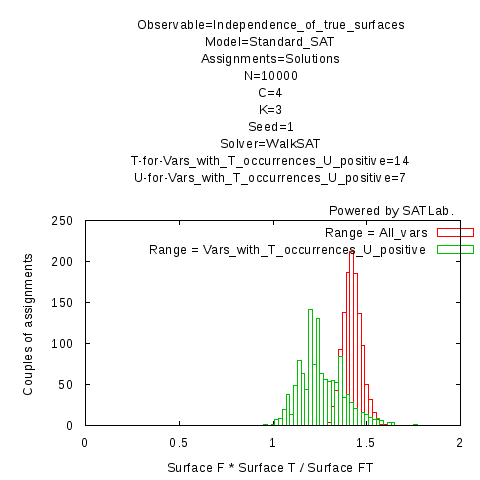}
\par\end{centering}

\caption{\label{fig:SATLab-Non-independence-T-U}The distributional model slightly
curbs the non-independence of surfaces.}

\end{figure}

\section{Conclusion on the Second Moment Method}

Contrary to the First Moment Method, which works more or less finely,
the Second Moment Method will not always work. Furthermore, it is
rather difficult to make it work, and we were able to make it work
only under very artificial conditions with respect to reality. Moreover,
even when it works, we have not been able to find strong lower bounds
with it. We got stuck at $2.83$ for 3 different models (standard
drawing model, implicants and standard distributional model).

However we did not prove that it is impossible to find better lower
bounds with our general framework, this is just numerical experiments.
Moreover our framework may not be perfect, perhaps the parameters
we consider are not relevant for the Second Moment Method, so there
is still hope in making the Second Moment Method work and give higher
lower bounds on the threshold of \threesat.

\bibliographystyle{plain}
\bibliography{../library}

\begin{thebibliography}{1}

\bibitem{Achlioptas2002}
Dimitris Achlioptas and Cristopher Moore.
\newblock {The asymptotic order of the random k-SAT threshold}.
\newblock In {\em The 43rd Annual IEEE Symposium on Foundations of Computer
  Science}, pages 779--788. IEEE Computer Society, 2002.

\bibitem{Achlioptas2004}
Dimitris Achlioptas and Yuval Peres.
\newblock {The Threshold for Random k-SAT is 2\^{}k ln2 - O(k)}.
\newblock {\em JAMS: Journal of the American Mathematical Society},
  17:947--973, 2004.

\bibitem{Boufkhad1999}
Yacine Boufkhad and Olivier Dubois.
\newblock {Length of prime implicants and number of solutions of random CNF
  formulae}.
\newblock {\em Theoretical Computer Science}, 215(1-2):1--30, 1999.

\bibitem{Diaz2009}
J.~D\'{\i}az, Lefteris~M. Kirousis, D.~Mitsche, and X.~P\'{e}rez-Gim\'{e}nez.
\newblock {On the satisfiability threshold of formulas with three literals per
  clause}.
\newblock {\em Theoretical Computer Science}, 410(30-32):2920--2934, 2009.

\bibitem{Friedgut1999}
Ehud Friedgut and J.~Bourgain.
\newblock {Sharp thresholds of graph properties, and the k-sat problem}.
\newblock {\em Journal of the American Mathematical Society}, 12(4):1017--1054,
  1999.

\bibitem{Hajiaghayi2003}
M.T. Hajiaghayi and G.B. Sorkin.
\newblock {The satisfiability threshold of random 3-SAT is at least 3.52}.
\newblock {\em IBM Research Report RC22942}, 2003.

\bibitem{Hugel2010a}
Thomas Hugel.
\newblock {SATLab}, 2010.

\bibitem{Kaporis2006}
A.C. Kaporis, Lefteris~M. Kirousis, and E.G. Lalas.
\newblock {The probabilistic analysis of a greedy satisfiability algorithm}.
\newblock {\em Random Structures and Algorithms}, 28(4):444--480, 2006.

\end{thebibliography}

\end{document}